\documentclass[a4paper,twocolumn,11pt,accepted=2020-11-05]{quantumarticle}
\pdfoutput=1
\usepackage[utf8]{inputenc}
\usepackage[english]{babel}
\usepackage[T1]{fontenc}
\usepackage{amsmath}
\usepackage{hyperref}

\usepackage{tikz}
\usepackage{lipsum}

\usepackage[numbers]{natbib} 
\usepackage{epsfig,amssymb,amsthm,mathrsfs,dsfont,ctable,url,commath}
\usepackage{color} \usepackage{bm} \usepackage{braket}
\usepackage{units}

\usepackage{amscd}
\usepackage{commath,mathtools,thmtools}
\usepackage{tensor}
\input{qcircuit}

\newtheorem{lemma}{Lemma} \newtheorem{proposition}{Proposition}

\newtheorem{corollary}{Corollary} \newtheorem{theorem}{Theorem}

 \newtheorem{definition}{Definition}
 \newtheorem*{definition*}{Definition}
\theoremstyle{remark}
 \newtheorem{remark}{Remark}

\def\>{\rangle}
\def\<{\langle} 
\def\trnsfrm#1{\mathcal #1}\newcommand\test[1]{\bm{#1}}
   \def\rA{{\rm A}}\def\rB{{\rm
    B}}\def\rC{{\rm C}}\def\rD{{\rm D}} \def\rE{{\rm E}} \def\rF{{\rm
    F}} \def\rG{{\rm G}} \def\rI{{\rm I}}   \def\rX{{\rm X}} \def\rY{{\rm Y}}
\def\rZ{{\rm Z}}\def\rU{{\rm U}}\def\rV{{\rm V}}\def\rW{{\rm W}}
\def\rN{{\rm N}}\def\rM{{\rm M}}
\def\tA{\trnsfrm A}\def\tB{\trnsfrm B}\def\tC{\trnsfrm
  C}\def\tD{\trnsfrm D}  \def\tI{\trnsfrm
  I}\def\tT{\trnsfrm T}\def\tU{\trnsfrm U} \def\tV{\trnsfrm V}
\def\tW{\trnsfrm W}
\def\tF{\trnsfrm
  F}\def\tE{\trnsfrm E} \def\tR{\trnsfrm R}
\def\tT{\trnsfrm T}
\def\Ref{{\mathsf{Ref}}} 
\def\Cntset{{\mathsf{Eff}}} 
\def\Stset{{\mathsf{St}}}

\def\Trnset{{\mathsf{Transf}}}
 
\def\K#1{\left|#1\right)}  \def\B#1{\left(#1\right|}
\def\SC#1#2{(#1| #2)}  

\def\Reals{{\mathbb R}}

\usepackage{enumitem}

\usepackage{array}
\newcolumntype{?}{!{\vrule width 1.2pt}}

\usepackage{soul}

\definecolor{blue1}{rgb}{0.03, 0.27, 0.49}

\usepackage{pifont}
\usepackage{xcolor}

\begin{document}

\title{Information and disturbance in operational probabilistic theories}

\author{Giacomo Mauro D'Ariano}
\email{dariano@unipv.it}
\affiliation{QUIT group,   Physics    Dept.,   Pavia
  University, and  INFN Sezione di Pavia,  via Bassi 6,
  27100 Pavia, Italy}
\orcid{0000-0003-0602-5519}

\author{Paolo Perinotti}
\email{paolo.perinotti@unipv.it}
\affiliation{QUIT    group,   Physics    Dept.,   Pavia
  University, and  INFN Sezione di Pavia,  via Bassi 6,
  27100 Pavia, Italy}
\orcid{0000-0003-4825-4264}

\author{Alessandro Tosini}
\email{alessandro.tosini@unipv.it}
\affiliation{QUIT group, Physics Dept., Pavia University, and INFN
  Sezione di Pavia, via Bassi 6, 27100 Pavia, Italy}
\orcid{0000-0001-8599-4427}

\maketitle

\begin{abstract}
  Any measurement is intended to provide {\em information} on a
  system, namely knowledge about its state. However, we learn from
  quantum theory that it is generally impossible to extract
  information without disturbing the state of the system or its
  correlations with other systems.  In this paper we address the issue
  of the interplay between information and disturbance for a general
  operational probabilistic theory.  The traditional notion of
  disturbance considers the fate of the system state after the
  measurement. However, the fact that the system state is left
  untouched ensures that also correlations are preserved only in the
  presence of local discriminability.  Here we provide the definition
  of disturbance that is appropriate for a general theory.  Moreover,
  since in a theory without causality information can be gathered also
  on the effect, we generalise the notion of no-information test. We
  then prove an equivalent condition for no-information without
  disturbance---{\em atomicity of the identity}---namely the
  impossibility of achieving the trivial evolution---the {\em
    identity}---as the coarse-graining of a set of non trivial ones.
  We prove a general theorem showing that information that can be
  retrieved without disturbance corresponds to perfectly repeatable
  and discriminating tests. Based on this, we prove a structure
  theorem for operational probabilistic theories, showing that the set
  of states of any system decomposes as a direct sum of perfectly
  discriminable sets, and such decomposition is preserved under system
  composition.  As a consequence, a theory is such that any
  information can be extracted without disturbance only if all its
  systems are classical. Finally, we show via concrete examples that
  no-information without disturbance is independent of both local
  discriminability and purification.
\end{abstract}

\section{introduction}

The possibility that gathering information on a physical system may
affect the state of the system itself was introduced by Heisenberg in
his famous gedanken experiment~\cite{Heisenberg1927}, which became the
first paradigm of quantum mechanics. The issue raised by Heisenberg
spawned a vaste literature up to present days
(see~\cite{BUSCH2007155,RevModPhys.86.1261} as recent reviews), with a
variety of quantifications of ``information'' and ``disturbance'' and
corresponding tradeoff
relations~\cite{PhysRevA.53.2038,d2003heisenberg,OZAWA2004350,PhysRevA.73.042307}. All
these results are quantitative accounts of a core issue in quantum
theory, the {\em no-information without disturbance}
theorem~\cite{Busch2009,dariano_chiribella_perinotti_2017}.  The
proofs of the theorem rely on the mathematical structure of quantum
theory, and thus do not emphasise the logical relation between
no-information without disturbance and other quantum features, such as
\emph{local discriminability} (the possibility of discriminating
multipartite states via only local measurements) or
\emph{purification} (every mixed state can be obtained as the marginal
state of a pure state).

The framework here used for exploring the relation between information
and disturbance is that of \emph{operational probabilistic theories}
(OPTs)~\cite{PhysRevA.81.062348,Chiribella2016,dariano_chiribella_perinotti_2017}. In
this setting a rigorous formulation of the notions of system, process,
and their compositions is given, which constitutes the grammar for the
probabilistic description of an experiment. Quantum theory and
classical theory are two instances of OPTs.

For some probabilistic theories which can be reframed as OPTs, the
definitions of information and disturbance have been investigated in
the presence of local discriminability, purification, and
causality~\cite{barrett2007information,KIMURA2010175,BARNUM20113,Heinosaari2019nofreeinformation}.
For OPTs satisfying those three axioms the no-information without
disturbance theorem has been proved in
Refs.~\cite{PhysRevA.81.062348,dariano_chiribella_perinotti_2017}.  In
the present paper we point out a weakness in the existing notion of
disturbance, which is ubiquitous in all past approaches. Indeed, the
conventional definition of disturbance asserts that an experiment does
not disturb the system if and only if its overall effect is to leave
unchanged the states of the system, disregarding the effects of the
experiment on the environment. Whilst this captures the meaning of
disturbance within quantum theory, we cannot consistently apply the
same notion in theories that violate local discriminability. A
significative case is that of the Fermionic
theory~\cite{Bravyi2002210,d2014feynman,D_Ariano_2014} where, due to
the parity superselection rule, an operation that does not disturb a
bunch of Fermionic systems still could affect their correlations with
other systems. This issue can be cured asking a non-disturbing
experiment to preserve not only the system state, but also its
purifications~\cite{PhysRevA.81.062348,dariano_chiribella_perinotti_2017}. This
extension of the notion of disturbance is general enough to capture
the operational meaning of disturbance for Fermionic systems, however,
it is still unsatisfactory, since it cannot be used to describe
disturbance in models that do not enjoy purification, e.~g.~ classical
information theory.

Here we will define non-disturbing operations only by referring to the
OPT framework, thus providing a notion that holds also for theories
that do not satisfy local discriminability, purification, or
causality, and even for theories whose sets of states are not convex.
Given a system, and an operation on it, the fate of any possible
dilation of the states of the system after the operation is taken into
account, where by {\em dilation} we mean any state of a larger system
whose marginal is the dilated state~\footnote{We remind that for
  non-causal theories the marginal is not unique, hence more
  generally, we require that {\em one of the marginals} is the given
  state. }. Moreover, due to the lack of causality, effects and states
must be treated on the same footing, and we extend the notion of
information also encompassing the information about the output. We
prove then a necessary and sufficient condition for a theory to
satisfy no-information without disturbance. The condition is the
impossibility of realizing the identity transformation as a nontrivial
coarse-graining of a set of operations. Technically speaking the above
condition amounts to {\em atomicity of the identity}. Finally, since a
theory might satisfy no-information without disturbance only when
restricted to some collections of preparations and measurements, we
will provide a weaker necessary and sufficient condition for this
case.

Similarly to the Heisenberg uncertainty relations, no-information
without disturbance has been considered as a characteristic quantum
trait. Instead, as we will see here, this feature can be exhibited in
the absence of most of the principles of quantum
theory~\cite{dariano_chiribella_perinotti_2017}, and it is ubiquitous
among OPTs. Moreover, the most general case is that of an OPT where
{\em some} information can be extracted without disturbance, in which
case this information has all the features of a classical one.  On the
other hand, the only kind of systems that allow for extracting any
information without disturbance are classical systems.  This
observation provides an alternative way of characterising classical
systems with respect to Ref.~\cite{Pfister:2013aa}.

In Section~\ref{sec:II} we review the framework of operational
probabilistic theories and some relevant features that characterize
quantum theory within this scenario. In Section~\ref{sec:III}, after
introducing the definition of information and disturbance, we present
the main results of this paper: i) the atomicity of the identity
evolution as a necessary and sufficient condition for no-information
without disturbance; ii) other equivalent necessary and sufficient
conditions in terms of properties of reversible evolutions of the
theory; iii) we prove a structure theorem for theories where some
information can be extracted without disturbance; iv) we prove that
the information that can be extracted without disturbance is
``classical'', in the sense that its measurement is a repeatable
reading of shareable information; v) finally we prove that a theory in
which any information can be extracted without disturbance is a theory
where all systems are classical.  In Section~\ref{sec:IV} we
generalize the notion of {\em equality upon input} to general OPTs,
including the cases in which local discriminability does not
hold. Moreover, dealing also with non-causal theories, where states
and effects must be considered on the same footing, we introduce the
notion of {\em equality upon input and upon output}. This notion can
be used when only a subset of the preparations and of the measurements
are accessible, e.g.~in resource
theories~\cite{PhysRevX.9.031053,COECKE201659}. As a first application
we generalize the notion of information and disturbance to the
upon input and upon output scenario, providing a characterization of
the no-information without disturbance also in this case. In
Section~\ref{sec:VI} we deepen the relation between no-information
without disturbance and other characteristic properties of quantum
theory. We show that no-information without disturbance can be
satisfied independently of purification and local discriminability,
providing counterexamples based on some of the conditions mentioned
above and other conditions proved in this section.  We end with the
conclusions in Section~\ref{sec:VII}.

\section{The Framework}\label{sec:II}

In this section we review the framework of {\em operational
  probabilistic theories} (OPT) (we refer
to~\cite{PhysRevA.81.062348,dariano_chiribella_perinotti_2017,Chiribella2016}
for further details).

The primitives of an operational theory are the notions of
\emph{test}, \emph{event}, and \emph{system}. A test
$\{\tA_i\}_{i\in\rX}$ is the collection of events $\tA_i$, where $i$
labels the element of the outcome space $\rX$.  In the quantum case
$\tA_i$ is the $i$th quantum operation of the quantum instrument
$\{\tA_i\}_{i\in\rX}$. The notion of test bridges the experiment with
the theory, with $i\in\rX$ denoting the objective outcome, and $\tA_i$
the mathematical description of the corresponding event. The notion of
\emph{system}, here denoted by capital Roman letters
$\rA,\rB,\ldots$, rules connections of tests. 
An input and an output label are associated to any test (event). We represent a test
$\test{\tA}_\rX:=\{\tA_i\}_{i\in\rX}$ and its building events $\tA_i$ by the diagrams
\begin{equation*}
\Qcircuit @C=1em @R=.7em @! R { & \poloFantasmaCn \rA \qw &
  \gate{\test{\tA}_{\rX}} & \qw \poloFantasmaCn \rB &\qw}\;\ ,\;\ 
  \Qcircuit @C=1em @R=.7em @! R { & \poloFantasmaCn \rA \qw &
  \gate{\tA_i} & \qw \poloFantasmaCn \rB &\qw}\,,
\end{equation*}
respectively, with the rule that an output wire can be connected only
to an input wire with the same label. Thus, given two tests
$\test{\tA}_{\rX }$ and $\test{\tB}_{\rY }$ we can define their
\emph{sequential composition} $(\test{\tB\tA})_{\rX\times\rY }$ as the collection of events
\begin{equation*}
  \Qcircuit @C=1em @R=.7em @! R { & \poloFantasmaCn \rA \qw &
    \gate{\tB_j\tA_i} & \qw \poloFantasmaCn \rC &\qw}\;=\; 
  \Qcircuit @C=1em @R=.7em @! R {& \poloFantasmaCn{\rA}\qw & \gate{
      \tA_{i}} & \poloFantasmaCn{\rB}\qw &
    \gate{\tB_{j}}&\poloFantasmaCn {\rC}\qw&\qw}\,,
\end{equation*}
for $i\in\rX$ and $j\in\rY$.  A {\em singleton test} is a test
containing a single event. We call such an event {\em deterministic}.
For every system $\rA$ there exists a unique singleton test
$\{\tI_{\rA} \}$ such that $\tI_{\rB} \tA=\tA\tI_{\rA}=\tA$ for every
event $\tA$ with input $\rA$ and output $\rB$, and we call $\tI_\rA$
{\em identity} of system $\rA$.  Besides sequential compositions of
tests and events, a theory is specified by the rule for composing them
in parallel.  For every couple of systems $(\rA,\rB)$ we can form the
composite system $\rC:=\rA\rB$, on which we can perform tests
$(\test{\tC}\otimes \test{\tD})_{\rX\times\rY}$ with events
$\tC_i\otimes\tD_j$ in {\em parallel composition} represented as
follows
\begin{equation*}
  \begin{aligned}
    \Qcircuit @C=1em @R=.7em @! R {& \qw \poloFantasmaCn \rA &
      \multigate{1} {\tC_i\otimes \tD_j} & \qw
      \poloFantasmaCn \rB &\qw\\
      & \qw \poloFantasmaCn \rC & \ghost {\tC_i\otimes \tD_j} & \qw
      \poloFantasmaCn \rD &\qw}
  \end{aligned}\ =\ 
  \begin{aligned}
    \Qcircuit @C=1em @R=.7em @! R {& \qw \poloFantasmaCn \rA & \gate
      {\tC_i} & \qw \poloFantasmaCn \rB &\qw\\ & \qw \poloFantasmaCn
      \rC & \gate {\tD_j} & \qw \poloFantasmaCn \rD &\qw}
  \end{aligned}\,,
\end{equation*}
and satisfying the condition
$ ( \tE_{h} \otimes \tF_{k} ) ( \tC_{i} \otimes \tD_{j} ) = ( \tE_{h}
\tC_{i} ) \otimes ( \tF_{k}\tD_{j} ).  $
Notice that we use the tensor product symbol $\otimes$ for the
parallel composition rule.  Actually, for the quantum and the
classical OPT the parallel composition is the usual tensor product of
linear maps. However, for a general OPT, the parallel composition may
not coincide with a tensor product.

There exists a special system type $\rI$, the {\em trivial system},
such that $\rA\rI=\rI\rA=\rA$ for every system $\rA$. The tests with
input system $\rI$ and output $\rA$ are called {\em preparation tests}
of $\rA$, while the tests with input system $\rA$ and output $\rI$ are
called {\em observation tests} of $\rA$.  Preparation events of $\rA$
are graphically denoted as boxes without the input wire
$\Qcircuit @C=.5em @R=.5em { \prepareC{\rho} & \poloFantasmaCn{\rA}
  \qw & \qw }$
(or in formula as round kets $\K{\rho}_\rA$), and the
observation events by boxes with no output wire
\( \Qcircuit @C=.5em @R=.5em { & \poloFantasmaCn{\rA} \qw & \measureD{
    c} } \)
(in formula round bras $\B{ c}_\rA$).  For example, one can have
events of the following kind
\begin{equation*}
  \begin{aligned}
    \Qcircuit @C=1em @R=.7em @! R {&& \prepareC
      {\rho_i} & \qw \poloFantasmaCn \rB &\qw\\ & \qw \poloFantasmaCn
      \rC & \gate {\tD_j} & \qw \poloFantasmaCn \rD &\qw}
  \end{aligned}
\ =\ 
  \begin{aligned}
    \Qcircuit @C=1em @R=.7em @! R {    
      & \qw \poloFantasmaCn \rC & \gate{\rho_i\otimes \tD_j}& \qw
      \poloFantasmaCn {\rB\rD} &\qw}
  \end{aligned}\,.
\end{equation*}
We will always use the Greek letters to denote preparation tests
$\{\rho_i\}_{i\in\rX }$ and Latin letters to denote observation tests
$\{c_j\}_{j\in\rX }$ (we will not specify the system when it is clear
from the context).

  An arbitrary test obtained by parallel and sequential composition of
  box diagrams is called \emph{circuit}. A circuit is \emph{closed} if
  its overall input and output systems are trivial: it starts with a
  preparation test and ends with an observation test. An
  \emph{operational probabilistic theory} (OPT) is an operational
  theory where any closed circuit of tests corresponds to a
  probability distribution for the joint test. Compound tests from the
  trivial system to itself are independent, both for sequential and
  parallel composition, namely their joint probability distribution is
  given by the product of the respective joint probability
  distributions. For example the application of an observation event
  $c_i$ after the preparation event $\rho_j$ corresponds to the closed
  circuit $ \SC{ c_i }{ \rho_j}_\rA $ and denotes the probability of
  the outcome $(i,j)$ of the observation test $\test{c}_\rX$ after the
  preparation test $\test{\rho}_\rY$ of system $\rA$, i.e.
  \begin{equation*}
  \begin{aligned}
    \Qcircuit @C=.5em @R=.5em { \prepareC{\rho_j} &
      \poloFantasmaCn{\rA} \qw & \measureD{ c_i }
    }
  \end{aligned}:=
\Pr\Big[i,j\,\Big|\,
  \begin{aligned}
  \Qcircuit @C=.5em @R=.5em { \prepareC{\test{\rho}_\rY} & \poloFantasmaCn{\rA}
    \qw & \measureD{ \test{c }_\rX }} 
  \end{aligned}\,\Big].
  \end{equation*}
For a more complex example, consider the test
\begin{equation*}
\test{\tT}_\rU:=
    \begin{aligned}
\Qcircuit @C=1em @R=.7em @! R {
        \multiprepareC{2}{\test{\Psi}_{\rV}}&
\qw\poloFantasmaCn{\rA}&
\multigate{1}{\test{\tA}_\rW}&
\qw\poloFantasmaCn{\rB}&
\gate{\test{\tB}_\rX}&
\qw\poloFantasmaCn{\rC}&
\multimeasureD{1}{\test{E}_\rY}
\\
\pureghost{\test{\Psi}_\rV}&
\qw\poloFantasmaCn{\rD}&
\ghost{\test{\tA}_\rW}&\qw\poloFantasmaCn{\rE}&
\multigate{1}{\test{\tC}_\rZ}&
\qw\poloFantasmaCn{\rF}&
\pureghost{\test{E}_\rY}\qw
\\
\pureghost{\test{\Psi}_\rV}&
\qw&
\qw\poloFantasmaCn{\rG}&
\qw&\ghost{\test{\tC}_\rZ}
\\
      }
    \end{aligned}\,,
  \end{equation*}
with $\rU=\rV\times\rW\times\rX\times\rY\times\rZ$. Then we define
\begin{align*}
&  \begin{aligned}
    \Qcircuit @C=1em @R=.7em @! R { \multiprepareC{2}{\Psi_{i}}&
      \qw\poloFantasmaCn{\rA}& \multigate{1}{\tA_{j}}&
      \qw\poloFantasmaCn{\rB}& \gate{\tB_{k}}&
      \qw\poloFantasmaCn{\rC}& \multimeasureD{1}{E_{m}}
      \\
      \pureghost{\Psi_{i}}& \qw\poloFantasmaCn{\rD}&
      \ghost{\tA_{j}}&\qw\poloFantasmaCn{\rE}& \multigate{1}{\tC_{l}}&
      \qw\poloFantasmaCn{\rF}& \pureghost{E_{m}}\qw
      \\
      \pureghost{\Psi_{i}}& \qw& \qw\poloFantasmaCn{\rG}&
      \qw&\ghost{\tC_{l}}
      \\
    }
  \end{aligned}\\
&\qquad \qquad \qquad \qquad \qquad  :=\Pr[i,j,k,l,m|\test{\tT}_\rU].
\end{align*}

In the following, we will omit the parametric dependence on the
circuit if the latter is clear from the context.

Summarising: {\em by a closed circuit made of events we denote their
  joint probability upon the connection specified by the circuit
  graph, with nodes being the test boxes, and links being the system
  wires.}

Given a system $\rA$ of a probabilistic theory we can quotient the set
of preparation events of $\rA$ by the equivalence relation
$ \K{\rho}_\rA\sim\K{\sigma}_\rA \Leftrightarrow \SC{ c }{ \rho
}_\rA=\SC{ c }{ \sigma }_\rA$
for every observation event $c$. Similarly we can quotient
observation events.  The equivalence classes of preparation events and
observation events of $\rA$ will be denoted by the same symbols as
their elements $\K{\rho}_\rA$ and $\B{c}_\rA$, respectively, and will
be called \emph{state} and \emph{effect} for system $\rA$.
  
  For  every system $\rA$, we will denote by $\Stset(\rA)$, $\Cntset(\rA)$
  the sets of states and effects, respectively.  States and effects
  are real-valued functionals on each other, and can be
  naturally embedded in reciprocally dual real vector spaces,
  $\Stset_{\mathbb{R}}(\rA)$ and $\Cntset_{\mathbb{R}}(\rA)$, whose
  dimension $\dim(\rA)$ is assumed to be finite.

  In Appendix~\ref{sec:transformations} it is proved that an event
  $\tA$ with input system $\rA$ and output system $\rB$ induces a
  linear map from $\Stset_{\mathbb R}(\rA\rC)$ to
  $\Stset_{\mathbb R}(\rB\rC)$ for each ancillary system $\rC$. The
  collection of all these maps is called {\em transformation} from
  $\rA$ to $\rB$. More explicitly, 
  given two transformations $\tA,\tA^\prime\in\Trnset{(\rA,\rB)}$, one has
$\tA=\tA^\prime$, if and only if
\begin{equation*}
\begin{aligned}
  \Qcircuit @C=1em @R=.7em @! R {\multiprepareC{1}{\Psi}& \qw
    \poloFantasmaCn \rA &  \gate{\tA} & \qw \poloFantasmaCn \rB &\multimeasureD{1}{a} \\
    \pureghost\Psi &\qw& \qw \poloFantasmaCn \rC & \qw&\ghost{a}}
\end{aligned}
~=~
\begin{aligned}
\Qcircuit @C=1em @R=.7em @! R {\multiprepareC{1}{\Psi}& \qw
    \poloFantasmaCn \rA &  \gate{\tA^\prime} & \qw \poloFantasmaCn \rB &\multimeasureD{1}{a} \\
    \pureghost\Psi &\qw& \qw \poloFantasmaCn \rC & \qw&\ghost{a}}
\end{aligned}\ ,
\end{equation*}
for every $\rC$, every $\Psi\in\Stset{(\rA\rC)}$, and every
$a\in\Cntset(\rB\rC)$, namely they give the same probabilities within
every possible closed circuit. Notice that, using the fact that two
states (effects) are equal if and only if they give the same
probability when paired to every effect (state), the above
condition amounts to state that $\tA=\tA'$ if and only if
\begin{equation}
\label{eq:equaltransf}
\begin{aligned}
  \Qcircuit @C=1em @R=.7em @! R {
\multiprepareC{1}{\Psi}& \qw
    \poloFantasmaCn \rA &  \gate{\tA} & \qw \poloFantasmaCn \rB &\qw\\
    \pureghost\Psi &\qw& \qw \poloFantasmaCn \rC & \qw&\qw}
\end{aligned}
~=~
\begin{aligned}
\Qcircuit @C=1em @R=.7em @! R {
\multiprepareC{1}{\Psi}& \qw
    \poloFantasmaCn \rA &  \gate{\tA^\prime} & \qw \poloFantasmaCn \rB &\qw \\
    \pureghost\Psi &\qw& \qw \poloFantasmaCn \rC & \qw&\qw}
\end{aligned}\ ,
\end{equation}  
for every $\rC$, and every $\Psi\in\Stset{(\rA\rC)}$, or
\begin{equation}
\begin{aligned}
  \Qcircuit @C=1em @R=.7em @! R {&
    \qw\poloFantasmaCn \rA &  \gate{\tA} & \qw \poloFantasmaCn \rB  &\multimeasureD{1}{a}\\
    & \qw&\qw\poloFantasmaCn\rC & \qw&\ghost{a}}
\end{aligned}
~=~
\begin{aligned}
  \Qcircuit @C=1em @R=.7em @! R {&
    \qw\poloFantasmaCn \rA &  \gate{\tA^\prime} & \qw \poloFantasmaCn \rB  &\multimeasureD{1}{a}\\
& \qw&\qw\poloFantasmaCn\rC & \qw&\ghost{a}}
\end{aligned}
\end{equation}
 for every $\rC$, and every $a\in\Cntset{(\rB\rC)}$.

In the following, the symbols $\tA$ and
  $ \Qcircuit @C=.5em @R=.5em { & \poloFantasmaCn{\rA} \qw &
    \gate{\tA} & \poloFantasmaCn{\rB} \qw &\qw} $
  will be used to represent the transformation corresponding to
  the event $\tA$. The set of transformations from $\rA$ to $\rB$ will
  be denoted by $\Trnset(\rA,\rB)$, with linear span
  $\Trnset_\mathbb{R}(\rA,\rB)$. It is now obvious that a linear map
  $\tA\in\Trnset_{\mathbb R}(\rA,\rB)$ is {\em admissible} if it
  locally preserves the set of states $\Stset(\rA\rC)$, namely
  $\tA\otimes\tI_{\rC}(\Stset(\rA\rC)) \subseteq \Stset(\rB\rC)$, for
  every system $\rC$. In the following we will write $\tA\K{\Psi}_{\rA\rC}$
  instead of $\tA\otimes\tI_{\rC}\K{\Psi}_{\rA\rC}$, with
  $\Psi\in\Stset(\rA\rC)$ and $\tA\in\Trnset(\rA,\rB)$ when the
  domains are clear from the context.

  An operational probabilistic theory is now defined as a collection
  of systems and transformations with the above rules for parallel and
  sequential composition and with a probability associated to any
  closed circuit~\footnote{Notice that a more detailed account needs a
    category-theoretical definition of parallel and sequential
    composition of systems (see Ref.~\cite{Chiribella2016}).}.

We introduce now the notions of \emph{refinement} of an event and \emph{atomic} event.

\begin{definition}[Refinement of an event] A refinement of an event
  $\tC\in\Trnset(\rA,\rB)$ is given by a collection of events
  $\{\tD_i\}_{i\in\rX}$ from $\rA$ to $\rB$, such that there exists a
  test $\{\tD_i\}_{i\in\rY }$ with $X\subseteq Y$ and
  $\tC=\sum_{i\in\rX}\tD_i$.  We say that a refinement
  $\{\tD_i\}_{i\in\rX}$ of $\tC$ is trivial if $\tD_i=\lambda_i \tC$,
  $\lambda_i\in [0,1]$, for every $i\in\rX$. Conversely, $\tC$ is
  called the coarse-graining of the events $\{\tD_i\}_{i\in\rX}$,
  which we will also denote as $\tC=\tD_\rX$.
\end{definition}
In the following we will often refer to a
  refinement of $\tC$ simply as $\tC=\sum_{i\in\rX}\tD_i$, without
  specifying the test including the events $\tD_i$.

\begin{definition}[Refining event] Given two events
  $\tC,\tD\in\Trnset(\rA,\rB)$ we say that $\tD$ refines $\tC$, and
  write $\tD\prec\tC$, if there exist a refinement
  $\{\tD_i\}_{i\in\rX}$ of $\tC$ such that $\tD\in\{\tD_i\}_{i\in\rX}$.
\end{definition}

\begin{definition}[Non redundant test]
We call a test $\{\tA_i\}_{i\in\rX}$ non redundant when for every pair $i,j\in\rX$ one has 
$\tA_i\neq\lambda \tA_j$ for $\lambda>0$.
\end{definition}
Notice that a test that is redundant can be interpreted as a non
redundant test followed by a conditional coin tossing. As a
consequence a redundant test always gives some spurious information,
unrelated to the input state. From a redundant test one can achieve a
maximal non redundant one by taking the test made of coarse grainings
of all the sets of proportional elements.

\begin{definition}[Refinement set] Given an event
  $\tC\in\Trnset(\rA,\rB)$ we define its refinement set
  $\Ref_\tC$ the set  of all events that refine $\tC$.
\end{definition}

\begin{definition}[Atomic and refinable events]\label{def:atomic-events} An event $\tC$ is
  atomic if it admits only of trivial refinements, namely $\tD\prec\tC$
  implies $\tD=\lambda\tC$, $\lambda\in [0,1]$. An event is refinable
  if it is not atomic.
\end{definition}

In the special case of states, the word \emph{pure} is used as
synonym of atomic, with a pure state describing an event that provides
maximal knowledge about the system's preparation. This means that the
knowledge provided by a pure state cannot be further refined. As usual
a state that is not pure will be called \emph{mixed}. 

Another important relation between events is that of {\em coexistence}
and the consequent notion of {\em coexistent completion} for a set o
events.
\begin{definition}[Coexistent events and coexistent
    completion]\label{def:compatible-events} Two events
    $\tA,\tB\in\Trnset{(\rA,\rB)}$ are coexistent, and we write
    $\tA\wedge\tB$, if there exists a test
    $\{\tC_i\}_{i\in\rX}\subseteq\Trnset{(\rA,\rB)}$ such that
    $\tA=\tC_\rY$ and $\tB=\tC_\rZ$, where
    $\rY,\rZ\subseteq\rX$. Given an event $\tA$ we denote by
    $\widehat{\tA}$ the set of all events coexistent with $\tA$, and
    more generally, given a set of events $X$ its coexistent
    completion is defined as
\begin{equation}
  \widehat{X}=\{\tB;\tB\wedge\tA, \text{ for some } \tA\in X\}.
\end{equation}
\end{definition}

We observe that in the present general OPT framework features that seem
intuitive are not assumed, such as the convex completion of
transformations. A remarkable example is that of
``\emph{no-restriction of preparation tests}'' hypothesis, consisting
in the requirement that every collection of states that sum to a
deterministic state is a preparation test. Similarly, we do not assume
the \emph{no-restriction hypothesis for transformations}, namely the
requirement that every transformation that preserves the state set
belongs to a test.

A fundamental definition for this manuscript is that that of dilation.
\begin{definition}[Dilation]\label{def:dilation} We say that
  $\Psi\in\Stset(\rA\rB)$ is a dilation of $\rho\in\Stset(\rA)$ if
  \begin{equation*}
    \begin{aligned} \Qcircuit @C=1em @R=.7em @! R {\prepareC{\rho}&
        \qw \poloFantasmaCn \rA &\qw } 
    \end{aligned}
    ~=~
    \begin{aligned}
      \Qcircuit @C=1em @R=.7em @! R {\multiprepareC{1}{\Psi}& \qw \poloFantasmaCn \rA &\qw \\
        \pureghost{\Psi} & \qw \poloFantasmaCn \rB &\measureD {e}}
    \end{aligned}
  \end{equation*}
for some deterministic effect $e\in\Cntset(\rB)$. Analogously,
$c\in\Cntset(\rA\rB)$ is a dilation of $a\in\Cntset(\rA)$ if
\begin{equation*}
  \begin{aligned} \Qcircuit @C=1em @R=.7em @! R {& \qw
      \poloFantasmaCn \rA &\measureD{a} } 
\end{aligned}
  ~=~
\begin{aligned}
  \Qcircuit @C=1em @R=.7em @! R {&\qw \poloFantasmaCn \rA & \multimeasureD{1}{c} \\
    \prepareC{\omega}&\qw \poloFantasmaCn \rB &\ghost {c}}
\end{aligned}
\end{equation*}
for some deterministic state $\omega\in\Stset(\rB)$.  We denote by
$D_\rho$ the set of all dilations of the state $\rho$.  More
generally, given a collection of states $S\subseteq\Stset(\rA)$ we
define $D_S:=\bigcup_{\rho\in S }D_\rho$, with $D_{\Stset{(\rA)}}$
corresponding to the set of all states $\psi\in\Stset{(\rA\rB)}$ for
every system $\rB$. The same notation is used for the set of dilations
of effects.
  \end{definition}
  We remark that, given $\sigma\in S$, every state of the form
  $\sigma\otimes\rho$ belongs to $D_S$.

  Notice that there are generally more than one deterministic effect
  for the same system, differently from quantum theory, where the
  partial trace over the Hilbert space of the system is the only way
  to discard it. Instead, given a state $\Psi\in\Stset{(\rA\rB)}$, in
  a theory with more deterministic effects for the same system $\rB$
  the marginal state of system $\rA$ generally depends on the effect
  used to discard the system $\rB$. In the following we will call {\em
    marginal of a state with deterministic effect $e$} the specific
  marginal obtained by applying the effect
  $e\in\Cntset(\rB)$. Similarly, given an effect
  $c\in\Cntset{(\rA\rB)}$ its marginal of system $\rA$ depends on the
  choice of deterministic state on system $\rB$ and we will call {\em
    marginal of an effect with deterministic sate $\omega$} the
  specific marginal obtained by applying the deterministic state
  $\omega\in\Stset(\rB)$.

  Given a system $\rA$, in the dilations sets $D_{\Stset{(\rA})}$
  and $D_{\Cntset{(\rA})}$, there could be states and effects with the
  following property.
 
  \begin{definition}[Faithful state and faithful
    effect]\label{def:faith}
  A state $\Psi\in\Stset{(\rA\rC)}$ is faithful for system $\rA$ if given two
  transformations $\tA,\tA^\prime\in\Trnset{(\rA,\rB)}$, the condition
  $\tA\K{\Psi}_{\rA\rC} =\tA^\prime\K{\Psi}_{\rA\rC}$ implies
  $\tA=\tA^\prime$. Analogously, an effect
  $d\in\Cntset{(\rB\rC)}$ is faithful for $\rB$ if given two transformations
  $\tA,\tA^\prime\in\Trnset{(\rA,\rB)}$, the condition
  $\B{d}_{\rB\rC}\tA=\B{d}_{\rB\rC}\tA^\prime$ implies
  $\tA=\tA^\prime$.
\end{definition}

  \begin{remark}\label{rem:states-effects} 
    We observe that in the general framework, without further
    assumptions, states (preparations) and effects (measurements) are
    on equal footing, and any proposition proved for states can be
    proved in the same way for effects. Accordingly, since this paper
    relies only on the general framework of OPTs, all the results
    given in terms of states, dilations of states, and sets of
    dilations of states can be mirrored to results on effects,
    dilations of effects, and sets of dilations of effects,
    respectively. In the next Section~\ref{sec:axioms} we present some
    significant classes of OPTs that are obtained enriching the
    present framework with one or more properties, such as the
    possibility of performing the tomography of states using only
    local operations, or the possibility of obtaining an arbitrary
    mixed state as the marginal of a pure one. Among the properties
    discussed in the following there is also causality, which induces
    an asymmetry in the structure of states and effects of the
    theory. Indeed, as it happens in both classical and quantum
    theory, causality forces the existence of a unique deterministic
    effect, while the set of states typically presents several
    deterministic elements also in the presence of causality.
\end{remark}

\subsection{Relevant classes of OPTs}\label{sec:axioms}

A frequently highlighted property within the wider scenario of OPTs is
that of multipartite states discrimination via local measurements:

\begin{definition}[Local discriminability]\label{def:local-discriminability} It is possible to
  discriminate between any pair of states of composite systems using
  only local measurements. Mathematically, given two joint states
  $\Psi,\Psi^\prime\in\Stset(\rA\rB)$ with $\Psi\neq\Psi'$, there
  exist two effects $a\in\Cntset(\rA)$ and $b\in\Cntset(\rB)$, such
  that
\begin{equation*}
  \begin{aligned} \Qcircuit @C=1em @R=.7em @! R {\multiprepareC{1}{\Psi}& \qw \poloFantasmaCn \rA &\measureD a \\
      \pureghost\Psi & \qw \poloFantasmaCn \rB &\measureD b} \end{aligned}
  ~\not =~
\begin{aligned}
  \Qcircuit @C=1em @R=.7em @! R {\multiprepareC{1}{\Psi^\prime} & \qw \poloFantasmaCn \rA &\measureD a \\
      \pureghost{\Psi^\prime} & \qw \poloFantasmaCn \rB &\measureD b}
\end{aligned}\,.
\end{equation*}
\end{definition}
Notice that the names local discriminability and local tomography are
used interchangeably in the literature. Also in this manuscript we
will consider the two names as synonymous.

Two relevant consequences of local discriminability are: i) the local
characterization of transformations, stating that the local behaviour
of a transformation is sufficient to fully characterize the 
transformation itself; ii) the atomicity of parallel composition. 
Here we report those two features for the convenience of the reader.

\begin{proposition}[Local characterization of
  transformations]\label{cor:local-characterization} 
  If local discriminability holds, then for any two transformations
  $\tA,\tA^\prime\in\Trnset(\rA,\rB)$, the condition
  $\tA\K{\rho}_\rA=\tA^\prime\K{\rho}_\rA$ for every
  $\rho\in\Stset{(\rA)}$ implies $\tA=\tA^\prime$.
\end{proposition}
See Ref.~\cite{PhysRevA.81.062348} for the proof.

\begin{proposition}[Atomicity of parallel composition]\label{prop:atomic-parallel}
  If an OPT satisfies local discriminability then the parallel
  composition of atomic transformations is atomic.
\end{proposition}

For the proof of the above proposition see Ref.~\cite{DAriano:2014aa}.
We observe that an OPT with local discriminability allows for
tomography of multipartite states using only local measurements. In an
OPT with local discriminability, the linear space of effects of a
composite system is the tensor product of the linear spaces of effects
of the component systems, namely
$\Cntset{(\rA\rB)}_\Reals\equiv \Cntset{(\rA)}_\Reals\otimes
\Cntset{(\rB)}_\Reals$.
Thus, any bipartite effect $c\in\Cntset{(\rA\rB)}$ can be written as a
linear combination of product effects, and every probability
$\SC{c}{\rho}_{\rA\rB}$, for $\rho\in\Stset{(\rA\rB)}$, can be
computed as a linear combination of the probabilities
$(\B{a}_\rA\otimes\B{b}_\rB)\K{\rho}_{\rA\rB}$ arising from a finite
set of product effects. The same holds for the linear space of states
and in an OPT with local discriminability the parallel composition of
two states (effects) can be understood as a tensor product. Finally,
the relation $\dim{(\rA\rB)} =\dim (\rA)\dim(\rB)$ between the linear
dimension of the set of states/effects holds, whereas for theories
without local discriminability it holds
$\dim{(\rA\rB)} >\dim (\rA)\dim(\rB)$.

Recently it has been shown that relevant physical theories, such as
the Fermionic theory~\cite{Bravyi2002210}, can be described in the OPT
framework relaxing the property of local
discriminability~\cite{d2014feynman,D_Ariano_2014}. The most general
scenario for OPTs that exhibit a finite degree of holism is that of
OPTs with n-local discriminability for some
$n\in\mathbb{N}$~\cite{hardy2012limited}:

\begin{definition}[$n$-local discriminability]\label{def:ndisc} 
  A theory satisfies $n$-local discriminability if whenever two states
  $\rho$ and $\rho^\prime$ are different, there exist a $n$-local
  effect $b$ such that $\SC{b}{\rho}\neq \SC{b}{\rho^\prime}$. We say
  that an effect is n-local if it can be written as a conic
  combination of tensor products of effects that are at most
  n-partite.
\end{definition}

Two notable examples are indeed Fermionic quantum theory
and real quantum
computation~\cite{hardy2012limited,d2014feynman,D_Ariano_2014} that
are both 2-local tomographic.

Another relevant class of OPTs is that of theories with
purification~\cite{PhysRevA.81.062348,chiribella2011informational}. As
a result of this paper we will show
(Proposition~\ref{prop:purification-niwd-2}) that the set of convex
OPTs with purification is strictly smaller than the set of OPTs that
satisfy no-information without disturbance. Moreover, we will see that
a weak version of purification, which does not require the uniqueness
(as in quantum theory) but just the existence of a purification for
each state, is enough to imply no-information without disturbance
together with the convexity assumption. Accordingly, we define the
following class of OPTs.

\begin{definition}[States purification]\label{def:states-purification}
  We say that an OPT satisfies states purification if for every system
  $\rA$ and for every state $\rho\in\Stset(\rA)$, there exists a
  system $\rB$ and a pure state $\Psi\in\Stset(\rA\rB)$ which is a
  dilation of $\rho$.
\end{definition}

We will prove that also the analogous notion of purification for
effects, provided in the following, is sufficient to guarantee
no-information without disturbance.

\begin{definition}[Effects
  purification]\label{def:effects-purification}
  We say that an OPT satisfies effects purification if for every
  system $\rA$ and for every effect $a\in\Cntset(\rA)$, there exists a
  system $\rB$ and an atomic effect $c\in\Cntset(\rA\rB)$ that is a
  dilation of $a$.
\end{definition}

As already noticed, the above definitions do not require the
purification to be unique up to reversible transformations on the
purifying system.

The last relevant class of OPTs that we point out is that of causal
theories:
\begin{definition}[Causal OPTs]\label{def:causality} The probability
  of preparation events
  in a closed circuit is independent of the choice of observations.
\end{definition}
Mathematically, if
  $\{\rho_i\}_{i \in \rX} \subset \Stset(\rA)$ is a preparation test,
  then the conditional probability of the preparation $\rho_i$ given
  the choice of the observation test $\{a_j\}_{j \in \rY}$ is the
  marginal
  \begin{equation*} 
  \Pr\left(i |  \{a_j\}  \right) := \sum_{j\in\rY} \SC {a_j} {\rho_i}_\rA.
  \end{equation*}
  In a causal theory the marginal probability
  $\Pr\left(i | \{a_j\} \right)$ is independent of the choice of the
  observation test $\{a_j\}$: if $\{a_j\}_{j\in\rY }$ and
  $\{b_k\}_{k\in\rZ }$ are two different observation tests, then one
  has $\Pr \left(i | \{a_j\} \right)=\Pr \left(i | \{b_k\} \right)$.

  The present notion of causality is simply the {\em Einstein
    causality} expressed in the language of OPTs. As proved in
  Ref.~\cite{PhysRevA.81.062348} causality is equivalent to the
  existence a unique deterministic effect $e_\rA$. We call the effect
  $e_\rA$ \emph{the deterministic effect} for system $\rA$. By
  definiton in non-causal theories the deterministic effect cannot be
  unique.

\section{ Information and disturbance}\label{sec:III}
Within the general scenario of operational probabilistic theories, and
without further assumptions on the structure of the theory, we aim at
defining the notions of {\em non-disturbing} and {\em no-information
  test}. These notions have already been investigated for causal
theories (Definition~\ref{def:causality}) that satisfy local
discriminability (Definition~\ref{def:local-discriminability}) or
states purification (Definition~\ref{def:states-purification}). We
start highlighting the weakness of previous approaches in cases where
the above hypotheses do not hold. The disturbance and the information
produced by a test on a physical system $\rA$ are commonly defined in
relation to measurements and states of the system $\rA$ only,
disregarding the action of the same test on an enlarged systems
$\rA\rB$.

A test $\{\tA_i\}_{i\in\rX }$ on system $\rA$ is usually said to be
non-disturbing if for every $\rho\in\Stset{(\rA)}$ one has that
$\sum_i \tA_i\K{\rho}_{\rA}=\K{\rho}_{\rA}$. However, this definition
is not operationally consistent if applied to theories without local
discriminability. A physically relevant example is that of Fermionic
theory~\cite{Bravyi2002210} that, due to the parity superselection
rule, is non-local tomographic~\cite{d2014feynman,D_Ariano_2014} (it
is 2-local tomographic according to Definition~\ref{def:ndisc}). We
can see via a simple example that, for a Fermionic system $\rA$, a
test $\{\tA_i\}_{i\in\rX }$ such that
$\sum_i \tA_i\K{\rho}_{\rA}=\K{\rho}_{\rA}$ for every
$\rho\in\Stset{(\rA)}$ still can disturb the states of a composite
system $\rA\rB$.

The parity superselection rule on a system $\rN_F$ of $N$ Fermions
forbids any state corresponding to a superposition of vectors
belonging to $\mathcal{F}_N^e$ and $\mathcal{F}_N^o$, representing
Fock vector spaces with total even and odd occupation number,
respectively. As a consequence, the linearized set of states
$\Stset_\Reals{(\rN_F)}$ splits in the direct sum of two spaces,
containing the states with even and odd parity, respectively. It is
now convenient to make use of the projectors onto the well-defined
parity subspaces $\mathbb{P}_e$, for the even space, and
$\mathbb{P}_o$, for the odd one.  Notice that, since
$\mathbb{P}_e\mathbb{P}_o=\mathbb{P}_o\mathbb{P}_e=0$ any Fermionic
state $\rho$ will be of the form
$\rho=\mathbb{P}_e\rho \mathbb{P}_e+\mathbb{P}_o\rho \mathbb{P}_o$.
Consequently the {\em parity test}
$\{\mathbb{P}_e\cdot\mathbb{P}_e,\mathbb{P}_o\cdot\mathbb{P}_o\}$
leaves every state $\rho\in\Stset(\rN_F)$ unchanged.  Intuitively,
this seems to suggest that parity can be measured without
disturbing. Indeed, this view is in agreement with the notion of
disturbance that has been considered in the literature so far.
 
Consider now a mixed state $\rho\in \Stset{(\rN_F)}$, with
$\rho=p_e\rho_e+p_o\rho_o$, $\rho_e$ and $\rho_o$ an even and an odd
pure state respectively, and $p_e+p_o=1$.  For example, consider the states
\begin{equation*}
\rho_e=\ket{00}\!\!\bra{00},\;\rho_o=\ket{01}\!\!\bra{01},
\end{equation*}
and $p_e=p_o=1/2$, so that 
\begin{equation*}
\rho={\tfrac12}(\ket{00}\!\!\bra{00}+\ket{01}\!\!\bra{01}).
\end{equation*}

Since Fermionic theory allows for states
purification~\cite{d2014feynman} (see Definition
\ref{def:states-purification}), we can always find a state
$\Psi\in \Stset{(\rM_F)}$, with $M>N$ that purifies $\rho$.  Since
$\Psi$ is pure, it has a definite parity, say even. In our example one
can choose
\begin{equation}\label{PSI3}
\Psi=\frac 12 (\ket{000}+\ket{011})(\bra{000}+\bra{011}).
\end{equation}
Therefore, the local test on the system $\rN_F$ that measures
the parity of the system will not disturb the states of $\rN_F$ but
will decohere the state $\Psi$ to a mixed state, then introducing a
disturbance. For example, in our case
\begin{align*}
&(\mathbb{P}_e\otimes I)\Psi(\mathbb{P}_e\otimes I)+
(\mathbb{P}_o\otimes I)\Psi(\mathbb{P}_o\otimes I)\\=
&{\tfrac12}(\ket{000}\!\!\bra{000}+\ket{011}\!\!\bra{011}).
\end{align*}

In order to avoid the above issue, and to introduce a definition of
non-disturbing test that works also for theories without local
discriminability, one could say that a test $\{\tA_i\}_{i\in\rX }$ on
system $\rA$ is non-disturbing upon input of $\rho\in\Stset{(\rA)}$,
if for every $\sigma$ in the refinement set of $\rho$ and every
purification $\Psi_{\rA\rB}\in\Stset(\rA\rB)$ of $\sigma$ one has that
$\sum_i \tA_i\K{\Psi}_{\rA\rB}=\K{\Psi}_{\rA\rB}$. This route, which
has been proposed in
Refs.~\cite{PhysRevA.81.062348,dariano_chiribella_perinotti_2017},
captures the operational meaning of disturbance also for Fermionic
systems.  However, the definition of
Refs.~\cite{PhysRevA.81.062348,dariano_chiribella_perinotti_2017}
requires purification, and thus cannot be used in theories without
purification, e.~g.~the cases of PR boxes, or the classical theory of
information.

Based on the above motivations our proposal is to define the
disturbance (and the information) produced by a test in terms of its
action on dilations, both of states and effects. This leads to notions
of information and disturbance that are completely general and thus do
not depend on causality, local discriminability, or purification. This
will allow us to prove the no-information without disturbance theorem
for a very large class of OPTs. In this Section we first consider the
disturbance and the information provided by a test when no
restrictions are posed on the states and effects of the theory. The
generalization to a scenario where both preparations and measurements
are limited to given subsets is presented in
Section~\ref{sec:info-dist-XY}.

\begin{definition}[Non-disturbing test]\label{def:nodist}
  Consider a test $\{\tA_i\}_{i\in\rX }$ on system $\rA$. We say that
  the test is non-disturbing if
\begin{equation}
  \sum_i\tA_i= \tI _\rA.
\end{equation}
\end{definition}

Notice that, following the above definition, the test
$\{\tA_i\}_{i\in\rX}$ is disturbing if there exist
$\Psi\in D_{\Stset{(\rA)}}$, and $c\in D_{\Cntset{(\rA)}}$, such that
\begin{equation}
\sum_{i\in\rX}
\begin{aligned}
  \Qcircuit @C=1em @R=.7em @! R {\multiprepareC{1}{\Psi}& \qw
    \poloFantasmaCn \rA &  \gate{\tA_i} & \qw \poloFantasmaCn \rA &\multimeasureD{1}{c} \\
    \pureghost\Psi &\qw& \qw \poloFantasmaCn \rB & \qw&\ghost{c}}
\end{aligned}\neq
\begin{aligned}
  \Qcircuit @C=1em @R=.7em @! R {\multiprepareC{1}{\Psi}& \qw
    \poloFantasmaCn \rA  &\multimeasureD{1}{c} \\
    \pureghost\Psi &\qw \poloFantasmaCn \rB &\ghost{c}}
\end{aligned}\;.
  \end{equation}
  This definition of disturbance thus stresses the effect of a
  transformation on correlations with remote systems, indeed a test
  $\{\tA_i\}_{i\in\rX }$ is \emph{non-disturbing} if it is
  operationally equal to the identity transformation of system $\rA$,
  namely it acts as the identity on any possible state and effect of
  any composite system.

  \begin{remark} 
    We could have defined a non-disturbing test from $\rA$ to $\rC$ as
    follows
  \begin{equation}
    \sum_i\tA_i=\tR,
  \end{equation}
  where $\tR\in\Trnset(\rA,\rC)$ is a reversible transformation,
  namely there exists another transformation $\tW\in\Trnset(\rC,\rA)$
  such that $\tW\tR=\tI_\rA$, and $\tR\tW=\tI_\rC$. Indeed if the test
  provides a systematic reversible transformation on the inputs, then
  its effect can be trivially corrected by inverting it. The
  classification of non-disturbing test according to this definition
  is trivially provided by the classification according to
  Definition~\ref{def:nodist}. Indeed, the most general non-disturbing
  test from $\rA\to\rC$ is the sequence of tests of the form
  $\{\tA_i\tR\}_{i\in\rX}$, with $\{\tA_i\}_{i\in\rX}$ non-disturbing
  according to Definition~\ref{def:nodist}, and
  $\tR\in\Trnset(\rA,\rC)$ reversible.
\end{remark}

In the same spirit we can establish if a test provides
information. Again a test could provide information both on the input
(preparation) and on the output (observation).

Let us consider the task in which Bob wants to extract information on
states (effects) of Alice via the test
$\test{\tA}_\rX:=\{\tA_i\}_{i\in\rX }\subseteq\Trnset(\rA\rC)$.
Operationally, in the most general case the test $\test{\tA}_\rX$ can
be used in composite tests made of a preparation test
$\test{\Psi}_\rY:=\{\Psi_j\}_{j\in\rY}\subseteq\Stset(\rA\rB)$, the
test $\test{\tA}_\rX:=\{\tA_i\}_{i\in\rX }\subseteq\Trnset(\rA\rC)$
and an observation test
$\test{C}_\rZ:=\{c_k\}_{k\in\rZ}\subseteq\Cntset(\rC\rB)$, leading to
the joint probabilities
\begin{equation}\label{eq:ni-scheme}
\begin{aligned}
  \Qcircuit @C=1em @R=.7em @! R {\multiprepareC{1}{\Psi_j}& \qw
    \poloFantasmaCn \rA &  \gate{\tA_i} & \qw \poloFantasmaCn \rC &\multimeasureD{1}{c_k} \\
    \pureghost{\Psi_j} &\qw& \qw \poloFantasmaCn \rB & \qw&\ghost{c_k}}
\end{aligned}= p(j,i,k|\test{\Psi}_\rY, \test{\tA}_\rX,\test{C}_\rZ)
\end{equation}
associated with possible outcomes $j,i,k$, and where we explicitly
show the dependence of the joint probability distribution $p(j,i,k)$
on the tests composing the circuit.

Within this scenario, Bob can use both the test $\test{\tA}_\rX$ and
any observation test $\test{C}_\rZ$ in order to extract the
information on the inputs, while he can use both the test
$\test{\tA}_\rX$ and any preparation test $\test{\Psi}_\rY$ in order
to extract the information on the outputs. This leaves room for two
inequivalent conditions for a no-information test $\test{\tA}_\rX$.

  \begin{enumerate}[leftmargin=*]
  \item\label{strong} \emph{Strong condition for no-information test}.  
\begin{enumerate}[leftmargin=*]
\item \emph{No-information on inputs}\label{strongin}: the test $\test{\tA}_\rX$ is
  no-information on inputs if for every preparation test
  $\test{\Psi}_\rY$ and for every observation tests $\test{C}_\rZ$,
  the joint probability in Eq.~\eqref{eq:ni-scheme} factorizes as
  \begin{equation*}
\begin{aligned}
  p(j,i,k|\test{\Psi}_\rY,
  \test{\tA}_\rX,\test{C}_\rZ)\\
  = r(i|k;\test{\tA}_\rX,\test{C}_\rZ)s(j,k|\test{\Psi}_\rY,
  \test{\tA}_\rX,\test{C}_\rZ),
\end{aligned}
\end{equation*}
namely we impose that
$r(i|j,k;\test{\Psi}_\rY,
\test{\tA}_\rX,\test{C}_\rZ)=r(i|k;\test{\tA}_\rX,\test{C}_\rZ)$,
where the probability distribution $r$ does not depend on the
preparation test $\test{\Psi}_\rY$ (we remind that it may happen that
a probability distribution depends on a given test but not on its
outcomes). The interpretation of this condition is that the outcomes
of $\test{\tA}_\rX$ and their correlations with the outcomes of any
observation test do not provide information on the preparation.

\item \emph{No-information on outputs}:\label{strongout} the test
  $\test{\tA}_\rX$ is no-information on outputs if for every
  $\test{C}_\rZ$ and for every $\test{\Psi}_\rY$ the joint probability
  in Eq.~\eqref{eq:ni-scheme} factorizes as
  \begin{equation*}
\begin{aligned}
  p(j,i,k|\test{\Psi}_\rY,
  \test{\tA}_\rX,\test{C}_\rZ)\\
  =r(i|j;\test{\Psi}_\rY,\test{\tA}_\rX)s(j,k|\test{\Psi}_\rY,
  \test{\tA}_\rX,\test{C}_\rZ),
\end{aligned}
\end{equation*}
where the probability distribution $r$ does not depend on the
observation test $\test{C}_\rZ$. This condition ensures that the outcomes of
$\test{\tA}_\rX$ and their correlations with the outcomes of any
preparation test do not provide information on the observation.
\end{enumerate}

\item \label{weak}\emph{Weak condition for no-information test}. 
\begin{enumerate}[leftmargin=*]
\item \emph{No-information on inputs}: the test $\test{\tA}_\rX$ is
  no-information on inputs if for every preparation test
  $\test{\Psi}_\rY$ and for every observation test $\test{C}_\rZ$, the
  joint probability in Eq.~\eqref{eq:ni-scheme} is such that
  \begin{equation}\label{eq:no-info-states-weak}
    \begin{aligned}
  \sum_k p(j,i,k|\test{\Psi}_\rY,
  \test{\tA}_\rX,\test{C}_\rZ)\\
  =r(i|\test{\tA}_\rX,\test{C}_\rZ)s(j|\test{\Psi}_\rY, 
  \test{\tA}_\rX,\test{C}_\rZ),
\end{aligned}
\end{equation}
where the probability distribution $r$ does not depend on the
preparation test $\test{\Psi}_\rY$. The interpretation of this
condition is that the outcomes of $\test{\tA}_\rX$ do not provide
information on the preparation, whenever we ignore the outcome of the
observation test.\label{eq:no-info-states-0}

\item \emph{No-information on outputs}: the test $\test{\tA}_\rX$ is
  no-information on outputs if for every observation test
  $\test{C}_\rZ$ and for every preparation test $\test{\Psi}_\rY$, the
  joint probability in Eq.~\eqref{eq:ni-scheme} is such that
  \begin{equation}\label{eq:no-info-effects-weak}
\begin{aligned}
 \sum_j p(j,i,k|\test{\Psi}_\rY,
    \test{\tA}_\rX,\test{C}_\rZ)=\\
r(i|\test{\Psi}_\rY,\test{\tA}_\rX)s(k|\test{\Psi}_\rY, 
    \test{\tA}_\rX,\test{C}_\rZ),
\end{aligned}
\end{equation}
  where the probability distribution $r$ does not depend on the
  observation test $\test{C}_\rZ$. This means that the outcomes of $\test{\tA}_\rX$
  do not provide information on the observation, whenever we ignore
  the outcome of the preparation test.\label{eq:no-info-effects-0}
\end{enumerate}

\end{enumerate}

It is elementary to see
that~\ref{strongin}$\Rightarrow$\ref{eq:no-info-states-0}
(and~\ref{strongout}$\Rightarrow$\ref{eq:no-info-effects-0}), namely
the strong no-information condition implies the weak one, both on
inputs and outputs. In the literature on no-information without
disturbance in quantum theory the authors take the weak
notion~\ref{eq:no-info-states-0} of no-information test (only on
inputs since the quantum theory is causal).  Here, we also choose the
conditions~\ref{weak} as expressed in the next definition (the
equivalence between the weak conditions~\ref{weak} and
Definition~\ref{def:noinfo} is proved in Appendix \ref{App:no-info
  test}). The motivation for this choice is that if an OPT satisfies
no-information without disturbance according to conditions~\ref{weak},
then it also satisfies no-information without disturbance in the
strongest sense of conditions~\ref{strong} (see
Remark~\ref{rem:niwdsw} in the following).

\begin{definition}[No-information test]\label{def:noinfo}
  A test $\{\tA_i\}_{i\in\rX }$ with events $\tA_i\in\Trnset(\rA,\rC)$
  is a no-information test if, for every choice of deterministic
  effect $e_{\rC\rB}$ and deterministic state $\omega_{\rA\rB}$, there
  exists a deterministic effect $f_{\rA\rB}$ and a deterministic state
  $\nu_{\rC\rB}$ such that for every $i\in\rX$ one has
  \begin{align}\label{eq:no-info-states}
    \B{e}_{\rC\rB}\tA_i= p_i(e)\B{f}_{\rA\rB},\\\label{eq:no-info-effects}
    \tA_i \K{\omega}_{\rA\rB}=q_i(\omega)
    \K{\nu}_{\rC\rB}.
  \end{align}
\end{definition}

According to Eq.~\eqref{eq:no-info-states} (that coincides with
Eq.~\eqref{eq:no-info-states-weak} in the weak condition of
item~\ref{eq:no-info-states-0}), the test $\{\tA_i\}_{i\in\rX }$ does
not provide information upon any possible input state.  However, the
probability distribution $p_i(e)$ might in principle provide
information about the effect $e$. On the other hand according to
Eq.~\eqref{eq:no-info-effects} (that coincides with
Eq.~\eqref{eq:no-info-effects-weak} in the weak condition of
item~\ref{eq:no-info-effects-0}), the test $\{\tA_i\}_{i\in\rX }$ does
not provide information upon output of any possible effect, while the
probability distribution $q_i(\omega)$ might in principle provide
information about the state $\omega$. The conjunction of the two
conditions implies that no-information is provided by the test about
$D_{\Stset{(\rA)}}$ and $D_{\Cntset{(\rA)}}$, namely about any
possible input state and output effect of any dilated system. The last
statement is proved in the following lemma.

\begin{lemma}\label{lem:no-info}
  Let the test $\{\tA_i\}_{i\in\rX }$ with events
  $\tA_i\in\Trnset(\rA,\rC)$ be a no-information test. Then one has
  \begin{align}\label{eq:no-info-states-lemma}
    \B{e}_{\rC\rB}\tA_i=r_i
    \B{f}_{\rA\rB},\\\label{eq:no-info-effects-lemma}
    \tA_i \K{\omega}_{\rA\rB}= r_i
    \K{\nu}_{\rC\rB}.
\end{align}
\end{lemma}
\begin{proof}
  By Eqs.~\eqref{eq:no-info-states} and~\eqref{eq:no-info-effects} one has
  \begin{align*}
    \B{e}_{\rC\rB}\tA_i\K{\omega}_{\rA\rB}= p_i(e)=q_i(\omega)=r_i,
\end{align*}
where we used the fact that $e,f$ and $\omega,\nu$ are respectively
deterministic effects and deterministic states. 
\end{proof}

\begin{remark}
  Notice that in Eq.~\eqref{eq:no-info-states} the probability of the
  transformation $\tA_i$ $\forall i\in\rX$ generally depends on the
  deterministic effect $e_{\rC\rB}$, this accounting for non-causal
  theories. In the more general case in which also the deterministic
  effect $f_{\rA\rB}$ on the right hand side of
  Eq.~\eqref{eq:no-info-states} depends on $i\in\rX$, the test
  $\{\tA_i\}_{i\in\rX}$ would provide information on the system state
  (this would happen, however, only for probabilistic states). An
  analogous argument holds for $\nu$ in
  Eq.~\eqref{eq:no-info-effects}.
\end{remark}

\subsection{No-information without disturbance}\label{sec:niwd}

In this section we state the condition of no-information without
disturbance and introduce criteria for it to be satisfied by an OPT.

\begin{definition}[OPT with no-information without disturbance]\label{def:niwd}
  We say that an OPT satisfies no-information without disturbance if,
  for every system $\rA$, and every test
  $\{\tA_i\}_{i\in\rX }\subseteq\Trnset(\rA)$, if the test is
  non-disturbing then it is a no-information test.
\end{definition}


\begin{theorem}
  An OPT satisfies no-information without disturbance if and only if
  the identity transformation is atomic for every system of the
  theory.
\label{thm:niwd}
\end{theorem}

\begin{proof}
  We start proving that if an OPT satisfies no-information without
  disturbance then the identity transformation is atomic. Consider a
  system $\rA$ of the theory, and a refinement $\{\tA_i\}_{i\in \rX}$
  ($\tA_i\in\Trnset(\rA)$ for every $i\in\rX$) of the identity map
  $\tI_\rA=\sum_i\tA_i$ for system $\rA$. The test
  $\{\tA_i\}_{i\in\rX}$ is clearly non-disturbing, therefore by
  hypothesis it is a no-information test. By definition of
  no-information test, and using Lemma~\ref{lem:no-info}, we know that
  for every deterministic effect $e_{\rC\rB}$, and deterministic state
  $\omega_{\rA\rB}$, there exists a deterministic effect $f_{\rA\rB}$
  and a deterministic state $\nu_{\rA\rB}$ such that for every
  $i\in\rX$ one has $\B{e}_{\rA\rB}\tA_i= r_i\B{f}_{\rA\rB}$, and
  $\tA_i \K{\omega}_{\rA\rB}=r_i \K{\nu}_{\rA\rB}$. Summing both sides
  of the last equation over the index $i\in\rX$, and remembering that
  $\sum_{i\in\rX}r_i=1$, we find that $e=f$ and
  $\omega=\nu$. Therefore, the no-information condition is
  \begin{align}\label{eq:cond11}
    \B{e}_{\rA\rB}\tA_i= r_i\B{e}_{\rA\rB},\\
\label{eq:cond12} 
\tA_i\K{\omega}_{\rA\rB}=r_i \K{\omega}_{\rA\rB},  
\end{align}
  for every deterministic effect $e_{\rA\rB}$ and for every
  deterministic state $\omega_{\rA\rB}$. Consider now an arbitrary
  pure state $\Psi\in\Stset{(\rA\rB)}$ (the same proof can be done
  choosing an arbitrary atomic effect $c\in\Cntset{(\rA\rB)}$). Since
  $\sum_i\tA_i\K{\Psi}_{\rA\rB}=\K{\Psi}_{\rA\rB}$, it follows that
\begin{align}\label{eq:cond2}
  \tA_i\K{\Psi}_{\rA\rB}=\lambda_i(\Psi)\K{\Psi}_{\rA\rB},\quad\sum_i \lambda_i(\Psi)=1,
\end{align} 
where the coefficients $\lambda_i(\Psi)$ generally depend on the state
$\Psi$. However, for each pure state $\Psi$ there exists a
deterministic effect $e^{\Psi}\in\Cntset{(\rA\rB)}$ such that
$\SC{e^{\Psi}}{\Psi}\neq 0$. Upon applying the deterministic effect
$e^{\Psi}$ on both sides of Eq.~\eqref{eq:cond2}, we get
\begin{align}
  (e^{\Psi}|\tA_i\K{\Psi}_{\rA\rB}=\lambda_i(\Psi)(e^{\Psi}|\Psi)_{\rA\rB}.
\end{align}
Now, applying both sides of Eq.~\eqref{eq:cond11} to $\Psi$, we get
\begin{align}
(e^\Psi|\tA_i\K{\Psi}_{\rA\rB}=r_i (e^\Psi|\Psi)_{\rA\rB},
\end{align}
and comparing the last two identities, considering that
$\SC{e^\Psi}{\Psi}_{\rA\rB}\neq0$, we obtain
\begin{align}
  \lambda_i(\Psi)=r_i,\qquad \forall i\in\rX.
\end{align}
Since this holds true for every pure state $\Psi$, we conclude that
$\lambda_i(\Psi)$ is independent of $\Psi$.  Then
$\tA_i\K{\rho}_{\rA\rB}=r_i\K{\rho}_{\rA\rB}$,
$\forall\rho\in\Stset(\rA\rB)$, proving that
$\tA_i=r_i\tI_\rA$. Notice that we implicitly assumed that the
probabilities $r_i$ do not depend on the choice of the system
$\rB$. Actually this can be proven as shown in Appendix~\ref{App:pulci}.

The converse implication, namely that if in an OPT the identity
transformation is atomic then a non-disturbing test is no-information,
is trivial.
\end{proof}

\begin{remark}\label{rem:niwdsw}
  Eq.~\eqref{eq:ni-scheme} shows the most general scenario in which a
  test $\{\tA_i\}_{i\in\rX}$ can be used to extract information on its
  inputs or on its outputs. We noticed that two inequivalent
  definitions of no-information tests are possible, a strong
  condition~\ref{strong}, and a weak condition~\ref{weak}, depending
  on the features of the joint probability distribution $p(j,i,k)$ of
  Eq.~\eqref{eq:ni-scheme}. However, due to the above theorem, if a
  theory satisfies no-information without disturbance in the weak
  sense, then a non-disturbing test
  $\{\tA_i\}_{i\in\rX}\in\Trnset{(\rA)}$ is such that
  $\tA_i=q_i\tI_\rA$, with $\sum_i q_i=1$. It follows that in
  Eq.~\eqref{eq:ni-scheme} the joint probability distribution
  $p(j,i,k)$ is of the form $p(j,i,k)=q_ip(j,k)$, and the test is also
  no-information in the strong sense.
\end{remark}

Besides the atomicity of the identity, we can provide other two
equivalent necessary and sufficient conditions for no-information
without disturbance.

\begin{proposition}\label{prop:niwd-rev}
  An OPT satisfies no-information without disturbance if and only if
  for every system there exists an atomic transformation which is
  either left- or right-reversible.
\end{proposition}
\begin{proof} We start proving that a theory with an atomic reversible
  transformation for each system satisfies no-information without
  disturbance. Let $\tR\in\Trnset(\rA,\rC)$ be atomic and
  left-reversible (the right-reversible case is analogous). Then
  consider a refinement $\tI_\rA=\sum_i\tA_i$, with
  $\tA_i\in\Trnset(\rA)$ for $i\in\rX $, of the identity
  transformation. By definition of identity map we have that
  $\tR\tI_\rA=\sum_i \tR\tA_i=\tR$, and due to the atomicity of $\tR$
  it must be $\tR\tA_i\propto\tR$ for every $i\in\rX $. Since $\tR$ is
  left-reversible (namely there exists $\tW\in\Trnset(\rC,\rA)$ such
  that $\tW\tR=\tI_\rA$) it follows that $\tA_i\propto\tI_\rA$ for
  every $i\in\rX $, which proves the atomicity of $\tI_\rA$.

The other implication, that in a theory that satisfies no-information
without disturbance for every system there exists an atomic
transformation which is either left- or right-reversible, is
trivial. Indeed, in a theory that satisfies no-information without
disturbance the identity, which is both right- and left-reversible, is
atomic as proved in Theorem~\ref{thm:niwd}.
\end{proof}

\begin{proposition}
  An OPT satisfies no-information without disturbance if and only if for every 
  system every reversible transformation is  atomic.
\end{proposition}
\begin{proof}
  We prove that if the theory satisfies no-information without
  disturbance, then every reversible transformation is atomic.
  Indeed, let $\tR\in\Trnset(\rA)$ be reversible, and suppose that
  $\tR=\sum_{i\in\rX}\tR_i$ for test $\{\tR_i\}_{i\in\rX}$.  Then, one
  has
\begin{equation}
\sum_{i\in\rX}\tR_i\tR^{-1}=\tI_\rA,
\end{equation}
and by Theorem~\ref{thm:niwd} one has that
$\tR_i\tR^{-1}=p_i\tI_\rA$. Finally, multiplying by $\tR$ to the
right, we conclude that $\tR_i=p_i\tR$, namely the refinements of
$\tR$ must be trivial. For the converse, it is sufficient to observe
that the identity is reversible.
\end{proof}

\subsection{Information without disturbance}\label{sec:V}

In this section we provide the general structure of the state spaces
and effect spaces of any theory where some information can be extracted
from a system without introducing disturbance. Such information is
``classical'' in the sense that the measurement is the reading of
information that is repeatable and shareable. In particular, for the
classical OPT the whole information encoded on a system can be read in
this way.  The proof of the above statements are based on the
following theorem.

\begin{theorem}\label{Thm:PROJ} The non redundant atomic refinement of
  the identity is unique for every
  system. Moreover, given the non redundant atomic refinement
  $\{\tA_i\}_{i\in\rX}\subseteq\Trnset(\rA)$ of the identity
  $\tI_\rA=\sum_i\tA_i$, one has $\tA_i\tA_j=\tA_i\delta_{ij}$.
\end{theorem}
\begin{proof} Suppose that the identity transformation of system $\rA$
  allows for two atomic refinements $\tI_\rA=\sum_{i\in\rX}\tA_i$, and
  $\tI_\rA=\sum_{j\in\rY}\tB_j$.  Since $\sum_i\tA_i\tB_j=\tB_j$, from
  the atomicity of the transformations $\tB_j$ we get
  $\tA_i\tB_j=c_{ij} \tB_j$, for some $c_{ij}\geq 0$ such that
  $\sum_{i\in \rX}c_{ij}=1$ $\forall j\in\rY$. Similarly we get
  $\tA_i\tB_j=d_{ij}\tA_i$ for some $d_{ij}\geq 0$ such that
  $\sum_{j\in \rY}d_{ij}=1$ $\forall i\in\rX$.  Then
  $c_{ij}\tB_j=d_{ij}\tA_i$.  By non redundancy one has that for fixed
  $j$ there is only one value of $i=i(j)$ such that $c_{ij}>0$, and
  normalisation gives $c_{i(j)j}=1$. By a similar argument for a fixed
  $i$ there is $j(i)$ such that $d_{ij(i)}=1$. Then one has
  $\tB_j=\tA_{i(j)}$. This proves the uniqueness of the non redundant
  atomic refinement of the identity.

  By the same argument as before, for the non redundant atomic
  refinement of the identity one has
  $\tA_i\tA_j=c_{ij} \tA_j=d_{ij}\tA_i$, for some
  $c_{ij},d_{ij}\geq 0$ such that
  $\sum_{i\in \rX}c_{ij}=\sum_{j\in \rX}d_{ij}=1$ $\forall i,j\in\rX$.
  By atomicity and non redundancy one must have
  $c_{ij}=d_{ij}=\delta_{i,j}$.
\end{proof}

The above theorem has as a consequence the following structure theorem
for OPTs.

\begin{corollary}\label{cor:struct}
  For any pair of systems $\rA$, $\rB$ of an OPT one has the following
  decomposition of the set of states and of the set of effects of
  $\rA\rB$
\begin{equation}\label{bigoplus}
\begin{aligned}
  \Stset(\rA\rB)&=\bigoplus_{(i,j)\in\rX\times\rY}\Stset_{ij}(\rA\rB),\\
  \Cntset(\rA\rB)&=\bigoplus_{(i,j)\in\rX\times\rY}\Cntset_{ij}(\rA\rB),
\end{aligned}
\end{equation}
where for non redundant atomic decompositions $\{\tA_i\}_{i\in\rX}$,
$\{\tB_j\}_{j\in\rY}$ of the identities $\tI_\rA$ and $\tI_\rB$, one
has
\begin{equation}\label{BIGOPS}
\begin{aligned}
(\tA_i\otimes\tB_j)\K{\Psi_{i'j'}}&=\delta_{ii'}\delta_{jj'}\K{\Psi_{ij}},\\
\B{c_{i'j'}} (\tA_i\otimes\tB_j)&=\delta_{ii'}\delta_{jj'}\B{c_{ij}},
\end{aligned}
\end{equation}
for all $\Psi_{i'j'}\in\Stset_{i'j'}(\rA\rB)$ and
$c_{i'j'}\in\Cntset_{i'j'}(\rA\rB)$.
\end{corollary}

\begin{remark}
  Notice that from Eq.~\eqref{bigoplus} it trivially follows that for
  any system $\rA$ the block decomposition holds
\begin{equation}\label{eq:bigoplus2}
  \Stset(\rA)=\bigoplus_{i\in\rX}\Stset_{i}(\rA),\quad \Cntset(\rA)=\bigoplus_{i\in\rX}\Cntset_{i}(\rA).
\end{equation}
However, Eq.~\eqref{bigoplus} contains the additional information that
the decomposition holds in that specific form also for composite
systems. This is not a straightforward consequence of the
decomposition of local states and local effects, as witnessed by the
Fermionic case. Indeed, the state in Eq.~\eqref{PSI3} does not have
definite parity for the two subsystems corresponding to two Fermions
on the left and one on the right, hence the state space cannot be of
the form in Eq.~\eqref{bigoplus}.
\end{remark}

\begin{remark}
  For a theory without atomicity of parallel composition it is
  possibile that the refinement $\tA_i\otimes\tB_j$ in
  Eq.~\eqref{BIGOPS} of $\tI_{\rA\rB}$ is not atomic. In such a case
  one has $\Stset(\rA\rB)=\bigoplus_{k\in\rZ}\Stset_k(\rA\rB)$, and
  $\Stset_{ij}(\rA\rB)=\bigoplus_{k\in\rZ_{ij}}\Stset_k(\rA\rB)$, for
  some partition $\rZ_{ij}$ of $\rZ$.
\end{remark}

\subsubsection{Full information without disturbance}

In the following we formalise the fact that a theory where any
information can be extracted via a non-disturbing test must have only
classical systems. Let us first define the notion of full-information
without disturbance.

\begin{definition}[Full-information without
  disturbance]\label{def:fiwd}
  An OPT satisfies full-information without disturbance if for every
  system $\rA$ and every test
  $\{\tB_j\}_{j\in\rY}\subseteq\Trnset{(\rA)}$ there exists a non
  disturbing test $\{\tA_i\}_{i\in\rX}\subseteq\Trnset{(\rA)}$ (namely
  $\sum_i\tA_i=\tI_\rA)$ such that
\begin{equation}\label{eq:fiwd}
  \tB_j=\sum_i p(j|i)\,\tV_{ij}\tA_i \tR_{ij},
\end{equation}  
for some probability distribution $p$ and reversible transformations
$\tV_{ij},\tR_{ij}\in\Trnset{(\rA)}$.
\end{definition}

As a consequence of the above definition we have the following lemma
on the structure of atomic maps in a theory with full-information
without disturbance.
\begin{lemma}\label{lem:fiwd-atomic}
  Consider an OPT with full-information without disturbance. Any
  atomic transformation $\tB\in\Trnset{(\rA)}$ is of the form
\begin{equation}\nonumber
\tB=\lambda\,\tU\tA_k=\lambda\,\tA_l\tU,
\end{equation}
where $\tA_k,\tA_l$ are atomic transformations in the unique
non-redundant refinement of the identity $\tI_\rA=\sum_i\tA_i$ of
Theorem \ref{Thm:PROJ}, $\tU\in\Trnset{(\rA)}$ is a reversible
transformation and $\lambda\geq 0$.
\end{lemma}
\begin{proof}
  Consider a test $\{\tB_j\}_{j\in\rY}\subseteq\Trnset{(\rA)}$ such
  that $\tB=\tB_j$ for some $j\in\rY$. Due to full-information without
  disturbance (see Eq.~\eqref{eq:fiwd} in Definition~\ref{def:fiwd}
  where we now take as non disturbing test $\{\tA_i\}_{i\in\rX}$ the
  unique non-redundant atomic refinement of the identity
  $\tI_\rA=\sum_i\tA_i$) and to the atomicity of $\tB$ one has
  $\tB=\lambda\,\tV\tA_i \tR$, for some $i\in\rX$, $\lambda\geq 0$ and
$\tV,\tR\in\Trnset{(\rA)}$ reversible transformations. Consider
now the three tests $\{\tV\tA_i\tR\}_{i\in\rX}$,
$\{\tA_i \tV\tR\}_{i\in\rX}$ and $\{\tV \tR\tA_i\}_{i\in\rX}$, and
observe that
$\sum_i \tV\tA_i\tR=\sum_i \tV\tR\tA_i=\sum_i \tA_i\tV\tR=\tV\tR=\tU$,
with $\tU\in\Trnset{(\rA)}$ reversible. We then conclude the proof
noticing that by Theorem~\ref{Thm:PROJ} the non redundant atomic
refinement of a reversible transformation is unique.
\end{proof}

We can now state the main Theorem of this section.
\begin{theorem}\label{thm:class}
  If an OPT is full-information without disturbance then every system
  of the theory is classical.
\end{theorem}

\begin{proof} 
  Consider an arbitrary system $\rA$ of the theory. Since by
  hypothesis the identity is not atomic, let $\{\tA_i\}_{i\in\rX}$ be
  the unique non-redundant atomic refinement of the identity
  $\tI_\rA=\sum_i\tA_i$ of Theorem \ref{Thm:PROJ}. Due to Corollary
  \ref{cor:struct} (and the immediately following remark) the sets of
  states and effects decompose as in Eq.~\eqref{eq:bigoplus2}. We now
  prove that all the blocks in such decompositions must be
  one-dimensional. To this end we show that any pair of states
  $\rho,\rho^\prime\in\Stset_i(\rA)$ is such that
  $\rho^\prime\propto \rho$.

  First we show that if an OPT satisfies full-information without
  disturbance then for every non null atomic $\rho\in\Stset_i{(\rA)}$
  and atomic $a\in\Cntset_i{(\rA)}$, one has $\SC{a}{\rho}\neq 0$.
  Given such $\rho\in\Stset_i{(\rA)}$ and $a\in\Cntset_i{(\rA)}$
  consider the transformation $\K{\rho}\B{a}\in\Trnset{(\rA)}$, that
  is generally not atomic. Due to Lemma~\ref{lem:fiwd-atomic} all the
  atomic refinements of the above transformation are of the form
  \begin{equation}\label{eq:meas-prep}
    \K{\rho}\B{a}= \sum_j\lambda_{j}\,\tU_j\tA_i,
\end{equation}
where each $\tA_i$ is the element of the non redundant refinement of
the identity $\tI_\rA$ such that $\tA_i\K{\rho}=\K{\rho}$, and
$\B{a}\tA_i=\B{a}$, $\lambda_j>0$, and the $\tU_j$ are reversible
trasformations. Applying both sides of Eq.~\eqref{eq:meas-prep} to the
state $\rho\in\Stset_i{(\rA)}$ one has
$\K{\rho}\SC{a}{\rho}= \sum_j\lambda_{j}\,\tU_j\K{\rho}$.  Reminding
that the $\tU_{j}$ are all reversible, $\lambda_j>0$ and $\rho$ is
non-null, one concludes that the right hand side cannot be null, and
this proves that also the pairing $\SC{a}{\rho}$ is non-null.

Let us apply the transformation in Eq.~\eqref{eq:meas-prep} to another
arbitrary atomic state $\rho^\prime\in\Stset_i{(\rA)}$. Since
$\tA_i\K{\rho^\prime}=\K{\rho^\prime}$ one finds
$\K{\rho}\SC{a}{\rho^\prime}=
\sum_j\lambda_{j}\,\tU_j\K{\rho^\prime}$,
for some $\lambda_j>0$ and $\tU_j$ reversible transformations. As
shown above it is $\SC{a}{\rho^\prime}\neq 0$, and using the atomicity
of $\rho$ one has $\K{\rho^\prime}\propto{\tU_j^{-1}}\K{\rho}$ for
every $j$. Since this holds true for every atomic state
$\rho^\prime\in\Stset_i{(\rA)}$, one has proved that all atomic states
(and then all states) in $\Stset_i{(\rA)}$ are proportional to the
same atomic state, let's say ${\tU_{j_0}^{-1}}\K{\rho}$ for some
$j_0$.  Via an analogous argument one can see that all effects in
$\Cntset_i{(\rA)}$ are proportional to the same atomic effect.
\end{proof}

\begin{remark}\label{rem:classical-systems}
  We remind that a system is classical when all its pure states are
  jointly perfectly discriminable. In this case the base of the conic
  hull of the pure states of each system is a simplex, which
  corresponds to a subset of the set of states for a convex theory.  A
  special case of theory whose systems are all classical is the usual
  classical information theory, where indeed one has full-information
  without disturbance. On the other hand, even when all systems are
  classical, the theory can differ from classical information theory
  e.~g.~in the rule for systems composition. For example there exist
  OPTs whose systems are all clssical but that do not satisfy local
  discriminability (see Ref.~\cite{PhysRevA.101.042118}).
\end{remark}

\section{Information and disturbance with restricted input and
  output}\label{sec:IV}

In this section we extend our previous results to study the relation
between disturbance and information when both input states and output
effects are limited to some given subsets. To this end we first
introduce the basics notion of identical transformations upon
restricted input and output resources.

\subsection{Operational identities between transformations}

As expressed in Eq.~\eqref{eq:equaltransf}, two transformations
$\tA,\tA^\prime\in\Trnset(\rA,\rB)$ of an OPT are said to be
operationally equal if for every system $\rC$ and for every state
$\Psi\in\Stset{(\rA\rC)}$ one has
$\tA\K{\Psi}_{\rA\rC}=\tA^\prime\K{\Psi}_{\rA\rC}$. However, two
non-identical maps $\tA,\tA^\prime\in\Trnset(\rA,\rB)$ could behave in
the same way when their action is restricted to a relevant subclass of
states.

The notion of identical transformation upon input of a state
$\rho\in\Stset(\rA)$ has been already introduced in the literature
(see
Refs.~\cite{nielsen_chuang_2010,dariano_chiribella_perinotti_2017} and
references therein):

\begin{definition}[Equal transformations upon input of
  $\rho$]\label{def:transformations-upon-input-of-rho}
  We say that two transformations $\tA,\tA^\prime\in\Trnset(\rA,\rB)$
  are equal upon input of $\rho\in\Stset(\rA)$, and write
  $\tA=_\rho\tA^\prime$, if for every $\sigma\in\Ref_\rho$ we have that
  $\tA \sigma = \tA^\prime\sigma$.
\end{definition}

\begin{remark}[Operational interpretation of equality upon input]
  The equality upon input of a state $\rho$ was originally introduced
  for quantum theory in Ref.~\cite{nielsen_chuang_2010}, where the
  authors extended the equality to the whole support of the chosen
  density matrix $\rho$. Within the OPT framework the equality upon
  input of $\rho$ is instead extended to the refinement set
  $\Ref_\rho$~\cite{dariano_chiribella_perinotti_2017}.  This choice
  can be easily motivated in operational terms: the equality
  $\tA=_\rho\tA^\prime$ means that the two maps $\tA$ and $\tA^\prime$
  are indistinguishable on the state $\rho$, independently of how it
  has been prepared. Suppose that the state $\rho$ is prepared by
  Alice as $\rho=\sum_{i\in\rX}\sigma_i$, for some refinement of
  $\rho$. Even Alice, using her knowledge of the preparation cannot
  distinguish between $\tA$ and $\tA^\prime$.
\end{remark}

From Proposition~\ref{cor:local-characterization} we know that the local
action of a map is sufficient to determine the map itself if the OPT
satisfies local discriminability (see Definition
\ref{def:local-discriminability}). However, for theories without local
discriminability the local action of a transformation might not be
sufficient to characterize it.  According to
Definition~\ref{def:transformations-upon-input-of-rho}, then, even if
$\tA=_\rho\tA^\prime$, still the maps $\tA$ and $\tA^\prime$ could act
differently upon input of dilations of $\rho$, namely it could be
$\tA\K{\Psi}_{\rA\rC}\neq \tA^\prime\K{\Psi}_{\rA\rC}$, for some
$\Psi\in \Ref_{D_{\rho}}$. In this case the difference between $\tA$
and $\tA^\prime$ would go undetected if their action on system $\rA$
only is considered. For this reason we introduce the notion of equal
transformations upon input of dilations of a state $\rho$.

\begin{definition}[Equal transformations upon
  input of $D_\rho$]\label{def:transformations-upon-input-of-D-rho}
  Given a state $\rho\in\Stset(\rA)$, we say that two transformations
  $\tA,\tA^\prime\in\Trnset(\rA,\rB)$ are equal upon input of
  $D_\rho$, and write $\tA=_{D_\rho}\tA^\prime$, if
  $\tA \K{\Psi}_{\rA\rC} = \tA^\prime\K{\Psi}_{\rA\rC}$ for every
  $\Psi\in \Ref_{D_{\rho}}$.
\end{definition}

Notice that the above definition requires that two transformations act
in the same way on the set $\Ref_{D_\rho}$.  Due to the absence of
no-restriction of preparation tests, it is not true in general that
$D_{\Ref_\rho}=\Ref_{D_\rho}$, as one might expect.  The only
inclusion that can be proved without further assumptions is
$\Ref_{D_\rho}\subseteq D_{\Ref_\rho}$, (see Lemma~\ref{lem:inclusion}
in Appendix~\ref{App:inclusion}).

Here we show that the two
Definitions~\ref{def:transformations-upon-input-of-rho} and
\ref{def:transformations-upon-input-of-D-rho} coincide for causal OPTs
with local discriminability. For this purpose we first need the
following lemma.

\begin{lemma}\label{lem:dilation} 
  In a causal OPT, if $\Psi\in\Stset{(\rA\rB)}$, with
  $\Psi\in \Ref_{D_\rho}$ for some $\rho\in\Stset{(\rA)}$, then
    \begin{equation}\label{Refrho}
      \{\B{b}_{\rB}\K{\Psi}_{\rA\rB}|b\in\Cntset{(\rB)}\}\subseteq\Ref_\rho.
\end{equation}
\end{lemma}
\begin{proof}
  Since $\Psi\in \Ref_{D_\rho}$ there exists a preparation test
  $\{\Psi,\bar\Psi,\Lambda\}\subseteq\Stset(\rA\rB)$ such that
  $\Psi+\bar\Psi\in D_\rho$. For an arbitrary $b\in\Cntset(\rB)$,
  thanks to causality that ensures the existence of a unique
  deterministic effect (see Definition \ref{def:causality}), one can
  construct the test
\begin{align*}
&\{\B{b}_\rB\K{\Psi}_{\rA\rB},\B{e-b}_\rB\K{\Psi}_{\rA\rB},
\\
&\B{b}_\rB\K{\bar\Psi}_{\rA\rB},\B{e-b}_\rB\K{\bar\Psi}_{\rA\rB},
\\
&\B{b}_\rB\K{\Lambda}_{\rA\rB},\B{e-b}_\rB\K{\Lambda}_{\rA\rB}\},
\end{align*}
where $e_\rB$ is the unique deterministic effect. Since the
coarse-graining of the first four elements is $\rho$, we conclude that
$\B{b}_\rB\K{\Psi}_{\rA\rB}\in\Ref_\rho$.
\end{proof}

\begin{proposition}\label{prop:local-characterization} 
  In a causal OPT with local discriminability, given two
  transformations $\tA,\tA^\prime\in\Trnset(\rA,\rB)$, the two
  conditions $\tA=_\rho \tA^\prime$ and $\tA=_{D_\rho} \tA^\prime$ are
  equivalent.
\end{proposition}
\begin{proof} We first prove that
  $\tA=_\rho \tA^\prime\Rightarrow \tA=_{D_\rho} \tA^\prime$. Consider
  an arbitrary $\Psi\in \Ref_{D_\rho}$, with
  $\rho\in\Stset{(\rA)}$. Let for example be
  $\Psi\in\Stset{(\rA\rC)}$. By hypothesis we have that
  $\tA=_{\rho}\tA^\prime$, namely
  $\tA\K{\sigma}_\rA=\tA^\prime\K{\sigma}_\rA$ for every
  $\sigma\in\Ref_\rho$. Then, due to Lemma~\ref{lem:dilation},
  $\forall b\in\Cntset{(\rB)}$, $\forall c\in\Cntset{(\rC)}$, we have
  that
  $\B{b}_\rB \B{c}_\rC (\tA\otimes\tI_\rC) \K{\Psi}_{\rA\rC}=\B{b}_\rB
  \B{c}_\rC (\tA^\prime\otimes\tI_\rC)\K{\Psi}_{\rA\rC}$,
  and by local discriminability (see Definition
  \ref{def:local-discriminability}) we conclude that
  $(\tA\otimes\tI_\rC)
  \K{\Psi}_{\rA\rC}=(\tA^\prime\otimes\tI_\rC)\K{\Psi}_{\rA\rC}$.
  Since this holds true for every $\Psi\in \Ref_{D_\rho}$, we conclude
  that $\tA=_{D_\rho}\tA^\prime$.  The converse implication
  $\tA=_{D_\rho} \tA^\prime\Rightarrow \tA=_\rho \tA^\prime$ is
  trivial.
\end{proof}

Dealing also with non-causal theories, where the role of states and
effects is interchangeable (see Remark~\ref{rem:states-effects}) it is
in order to introduce the counterpart of the above definition with
effects replacing states.  Accordingly we define equal transformations
upon output of dilations of an effect $a$.

\begin{definition}[Equal transformations upon output of
  $D_b$]\label{def:transformations-upon-output-of-D-a}
  Given an effect $b\in\Cntset(\rB)$, we say that two transformations
  $\tA,\tA^\prime\in\Trnset(\rA,\rB)$ are equal upon output of $D_b$,
  and write $\tA \tensor[_{D_b}]{=}{}\tA^\prime$, if
  $ \B{c}_{\rB\rC}\tA = \B{c}_{\rB\rC}\tA^\prime$ for every
  $c\in {\Ref_{D_b}}$.
\end{definition}

In the most general case one can define equality of two
transformations when both states and effects are limited to two given
subsets.

\begin{definition}[Equal transformations upon $X,Y$]\label{def:transformations-D-rho-a}
  We say that two transformations $\tA,\tA^\prime\in\Trnset(\rA,\rB)$
  are equal upon input of $X\subseteq D_{\Stset{(\rA)}}$ and upon
  output of $Y\subseteq D_{\Cntset{(\rB)}}$---or simply upon
  $X,Y$---and write $\tA\tensor[_Y]{=}{_X}\tA^\prime$, if
\begin{equation}
\begin{aligned}
  \Qcircuit @C=1em @R=.7em @! R {\multiprepareC{1}{\Psi}& \qw
    \poloFantasmaCn \rA &  \gate{\tA} & \qw \poloFantasmaCn \rB &\multimeasureD{1}{c} \\
    \pureghost\Psi &\qw& \qw \poloFantasmaCn \rC & \qw&\ghost{c}}
\end{aligned}=
\begin{aligned}
  \Qcircuit @C=1em @R=.7em @! R {\multiprepareC{1}{\Psi}& \qw
    \poloFantasmaCn \rA &  \gate{\tA^\prime} & \qw \poloFantasmaCn \rB &\multimeasureD{1}{c} \\
    \pureghost\Psi &\qw& \qw \poloFantasmaCn \rC & \qw&\ghost{c}}
\end{aligned}\;,
  \end{equation}  
  for every $\Psi\in\Ref_X$ and for every $c\in\Ref_Y$.
\end{definition}

As a special case, given a state $\rho\in\Stset(\rA)$, and an effect
$b\in\Cntset{(\rB)}$, two transformations
$\tA,\tA^\prime\in\Trnset(\rA,\rB)$ are equal upon $D_{\rho},D_b$ when
in the last definition we take $X=D_{\rho}$ and $Y=D_b$. Accordingly,
also Definitions~\ref{def:transformations-upon-input-of-D-rho} and
\ref{def:transformations-upon-output-of-D-a} are special cases of
Definition~\ref{def:transformations-D-rho-a}, with
$\tA=_{D_{\rho}}\tA^\prime$ corresponding to the choice
$X=D_{\rho}, Y=D_{\Cntset{(\rB)}}$ and
$\tA\tensor[_{D_b}]{=}{}\tA^\prime$ corresponding to the choice
$X=D_{\Stset{(\rA)}}, Y=D_b$.  Naturally, also the notion of equal
transformations $\tA=\tA^\prime$ is the one of
Definition~\ref{def:transformations-D-rho-a} with no restrictions on
the set states and effects, that is $X=D_{\Stset{(\rA)}}$,
$Y=D_{\Cntset{(\rB)}}$.

Based on the above identities of transformations, any property of an
OPT can be generalized in the upon $X,Y$ scenario. Here we only
present the case of two properties that will be used to derive the
following results.

In Definition~\ref{def:atomic-events} we introduced the notion of
atomic events, and we can provide a weaker version of the property of
atomicity for transformations.

\begin{definition}[Atomic and refinable transformation
  upon $X,Y$.]\label{def:atomicity-upon-rho-a}
  A transformation $\tA\in\Trnset(\rA,\rB)$ is atomic upon input of
  $X\subseteq D_{\Stset{(\rA)}}$ and upon output of
  $Y\subseteq D_{\Cntset{(\rB)}}$---or simply upon $X,Y$---if all its
  refinements are trivial upon $X,Y$, namely $\tB\prec\tA$ implies
  $\tB\tensor[_{Y}]{=}{_{X}}\lambda\tA$, $\lambda\in [0,1]$.
  Conversely, we say that an event is refinable upon $X,Y$ whenever it
  is not atomic upon $X,Y$
\end{definition}

Again, we mention as a special case the atomicity upon
$X=D_\rho,Y=D_b$ for some state $\rho\in\Stset{(\rA)}$ and effect
$b\in\Cntset{(\rB)}$, which in turn reduces to atomicity upon input of
$D_{\rho}$ (atomicity upon output of $D_{b}$) when
$X=D_\rho, Y=D_{\Cntset{(\rB)}}$ ($X=D_{\Stset{(\rA)}
},Y=D_b$).
Similarly, the usual notion of atomicity corresponds to the choice
$X=D_{\Stset{(\rA)}},Y=D_{\Cntset{(\rB)}}$.

Finally, the usual notion of faithful state and faithful effect (see
Definition~\ref{def:faith}) can be generalized to the notion of
faithful state and faithful effect upon input of $X$ and upon output
of $Y$.

\begin{definition}[Faithful state and faithful effect upon
  $X,Y$]\label{def:faithuponD}
  Consider two arbitrary transformations
  $\tA,\tA^\prime\in\Trnset{(\rA,\rB)}$ and two subsets
  $X\subseteq D_{\Stset{(\rA)}}$ and $Y\subseteq D_{\Cntset{(\rB)}}$.
  A state $\Psi\in\Stset{(\rA\rC)}$ is faithful upon input of $X$ and
  upon output of $Y$---or simply upon $X,Y$---if the condition
  $\tA\K{\Psi}_{\rA\rC}= \tA^\prime\K{\Psi}_{\rA\rC}$ implies
  $\tA\tensor[_{Y}]{=}{_{X}}\tA^\prime$.  Analogously, an effect
  $d\in\Cntset{(\rB\rC)}$ is faithful upon $X,Y$ if the condition
  $\B{d}_{\rB\rC}\tA=\B{d}_{\rB\rC}\tA^\prime$ implies
  $\tA\tensor[_{Y}]{=}{_{X}}\tA^\prime$.
\end{definition}

The case of faithful state and faithful effect of
Definition~\ref{def:faith} corresponds to the choice
$X=D_{\Stset{(\rA)}},Y=D_{\Cntset({\rB})}$ in the above definition.

\subsection{Information and disturbance upon
  X,Y}\label{sec:info-dist-XY}

We start generalizing Definition~\ref{def:nodist} of non-disturbing test:

\begin{definition}[Non-disturbing test upon $X,Y$]\label{def:nodist-XY}
  Consider a test $\{\tA_i\}_{i\in\rX }$ on system $\rA$. We say that
  the test is non-disturbing upon input of
  $X\subseteq D_{\Stset{(\rA)}}$ and upon output of
  $Y\subseteq D_{\Cntset{(\rA)}}$---or simply upon $X,Y$---if
  \begin{equation}\label{eq:non-dsturbing-ui-up}
    \sum_i\tA_i\tensor[_Y]{=}{_X} \tI _\rA.
\end{equation}
\end{definition}

According to Definition~\ref{def:nodist} a test $\{\tA_i\}_{i\in\rX }$
is \emph{non-disturbing} if it is operationally equal to the identity
transformation of system $\rA$; this is a special case of the above
definition when $X=D_{\Stset{(\rA)}}$, and $Y=D_{\Cntset{(\rA)}}$.

We can now determine if a test
$\{\tA_i\}_{i\in\rX }\subseteq\Trnset(\rA,\rC)$ provides information
upon $X\subseteq D_{\Stset(\rA)}$, $Y\subseteq D_{\Cntset(\rC)}$. We
first observe that the prescription upon $X,Y$ establishes that (i) we
are only interested in getting information on states in $X$ and on
effects in $Y$, and (ii) we can only use preparations that involve
states in $X$ and measurements that involve effects in $Y$ in order to
extract the information.  However, a test containing a state in $X$
(or an effect in $Y$) may involve events in $\widehat X$
($\widehat Y$), where the set $\widehat Z$ is the coexistent
completion of the set $Z$ as in
Definition~\ref{def:compatible-events}.

Let us focus on the scheme in Eq.  \eqref{eq:ni-scheme}, with the
state in the set $X$ and effect in the set $Y$.  A test that directly
provides information on $X$ may give different probability
distributions for different elements of $X$, given that one measures
an effect $c\in Y$. However, the test could also provide information
about $X$ indirectly. This is the case, for example, when for every
test $\{c_k\}_{k\in \rZ}\subseteq\widehat Y$, one has
$p(j,i|k)=p(i)p(j)$ for $\Psi_i\in X$, while factorisation does not
occur for $\Psi_i\in\widehat X\setminus X$. A similar situation can
occur when the information is about $Y$, with the roles of preparation
and measurement exchanged.

We thus generalize Definition~\ref{def:noinfo} of no-information test
as follows:
\begin{definition}[No-information test upon
  $X,Y$]\label{def:noinfo-XY}
  A test $\{\tA_i\}_{i\in\rX }$ with events $\tA_i\in\Trnset(\rA,\rC)$
  is a no-information test upon input of
  $X\subseteq D_{\Stset{(\rA)}}$ and upon output of
  $Y\subseteq D_{\Cntset{(\rC)}}$---or simply upon $X,Y$---if for
  every choice of deterministic effect $e_{\rC\rB}\in \widehat{Y}$ and
  deterministic state $\omega_{\rA\rB}\in \widehat{X}$, there exists a
  deterministic effect $f_{\rA\rB}$ and a deterministic state
  $\nu_{\rC\rB}$ such that for every $i\in\rX$ one has
  \begin{align}\label{eq:no-info-rho}
    \B{e}_{\rC\rB}\tA_i=_{\widehat{X}} p_i(e)
    \B{f}_{\rA\rB},\\\label{eq:no-info-a}
    \tA_i \K{\omega}_{\rA\rB}\tensor[_{\widehat{Y}}]{=}{} q_i(\omega)
    \K{\nu}_{\rC\rB}.
\end{align}
\end{definition}
As a special case the above definition coincides with
Definition~\ref{def:noinfo}, corresponding to $X=D_{\Stset{(\rA)}}$,
and $Y=D_{\Cntset{(\rA)}}$. Notice that in this case it is also
$\widehat{X}=X$ and $\widehat{Y}=Y$.

According to Eq.~\eqref{eq:no-info-rho}, the test
$\{\tA_i\}_{i\in\rX }$ does not provide information upon input of
$\widehat{X}$ once the observations are limited to $\widehat{Y}$
($e_{\rC\rB}\in \widehat{Y}$). However, the probability distribution
$p_i(e)$ might in principle provide information about the effect
$e$. On the other hand according to Eq.~\eqref{eq:no-info-a}, the test
$\{\tA_i\}_{i\in\rX }$ does not provide information upon output of
$\widehat{Y}$, once the preparations are limited to $\widehat{X}$
($\omega_{\rA\rB}\in \widehat{X}$), while the probability distribution
$q_i(\omega)$ might in principle provide information about the state
$\omega$. The conjunction of the two conditions implies that
no-information is provided by the test about $\widehat{X}$ and
$\widehat{Y}$:

\begin{lemma}\label{lem:noinfo-XY}
  Let the test $\{\tA_i\}_{i\in\rX }$ with events
  $\tA_i\in\Trnset(\rA,\rC)$ be a no-information test upon $X,Y$. Then
  one has
  \begin{align}
    \B{e}_{\rC\rB}\tA_i=_{\widehat{X}} r_i
    \B{f}_{\rA\rB},\label{eq:noinfo1}\\
    \tA_i \K{\omega}_{\rA\rB}\tensor[_{\widehat{Y}}]{=}{} r_i
    \K{\nu}_{\rC\rB}.
\end{align}
\end{lemma}
\begin{proof}
  By Eqs.~\eqref{eq:no-info-rho} and~\eqref{eq:no-info-a}, and
  remembering that $\omega_{\rA\rB}\in \widehat{X}$,
  $e_{\rC\rB}\in \widehat{Y}$, one has
  \begin{align*}
    \B{e}_{\rC\rB}\tA_i\K{\omega}_{\rA\rB}= p_i(e)=q_i(\omega)=r_i,
\end{align*}
where we used the fact that $e,f$ and $\omega,\nu$ are respectively
deterministic effects and deterministic states. 
\end{proof}

Here we state the condition of no-information without disturbance upon
$X,Y$:
\begin{definition}[No-information without disturbance upon $X,Y$]
  Consider a system $\rA$ and two subsets
  $X\subseteq D_{\Stset{(\rA)}}$, $Y\subseteq D_{\Cntset{(\rA)}}$.
  Then the OPT satisfies no-information without disturbance upon input
  of $X$ and upon output of $Y$---or simply upon $X,Y$---if for every
  test $\{\tA_i\}_{i\in\rX }\subseteq\Trnset(\rA)$ that is
  non-disturbing upon $X,Y$, the test does not provide information
  upon $X,Y$.
\end{definition}

Clearly the above definition generalizes Definition~\ref{def:niwd}
which corresponds to the choice $X=D_{\Stset{(\rA)}}$, and
$Y=D_{\Cntset{(\rA)}}$, namely the theory satisfies no-information
without disturbance and any informative test necessary disturbs.

All the results presented in Section~\ref{sec:niwd} on no-information
without disturbance can now be extended to the ``upon $X,Y$''
scenario. The result of Theorem~\ref{thm:niwd}, proving the atomicity
of the identity as a necessary and sufficient condition for
no-information without disturbance, is generalized by the following
theorem:

\begin{theorem}\label{thm:niwdu2}
  An OPT satisfies no-information without disturbance upon
  $\widehat{X},\widehat{Y}$, with $X\subseteq D_{\Stset{(\rA)}}$,
  $Y\subseteq D_{\Cntset{(\rA)}}$, if and only if the identity
  $\tI_\rA$ is atomic upon $\widehat{X},\widehat{Y}$.
\end{theorem}
\begin{proof}
  The proof follows the lines of that of Theorem~\ref{thm:niwd}.  In
  the present case Eq.~\eqref{eq:cond2} for pure state in
  $\widehat{X}$ holds upon $\widehat Y$.  Now, if we apply on both
  sides a deterministic effect $e_{\rA\rB}\in\widehat Y$, using
  Eq.~\eqref{eq:noinfo1} it may happen that $\SC{e}{\Psi}=0$. The case
  where this happens for every $e\in\widehat Y$ corresponds to a state
  $\Psi$ which is equal to the null state upon
  $\widehat{X},\widehat{Y}$. Considering the remaining cases one
  concludes that
  $\tA_i\K\rho \tensor[_{\widehat Y}]{=}{} r_i\K{\rho}$,
  $\forall\rho\in \widehat{X}$. Similarly, one concludes that
  $\B{b}\tA_i\tensor[]{=}{_{\widehat X}} r_i\B{b}$,
  $\forall b\in \widehat{Y}$.  These two last conditions imply that
  $\tA_i\tensor[_{\widehat Y}]{=}{_{\widehat X}} r_i\tI_\rA$,
  corresponding to atomicity of $\tI_\rA$ upon
  $\widehat{X}, \widehat{Y}$.  Also in this case the proof follows
  straightforwardly from that of Theorem~\ref{thm:niwd}. The opposite
  implication is trivial.
\end{proof}

Analogously, one can generalize the other
necessary and sufficient conditions of Section~\ref{sec:niwd}.  One
can also provide only sufficient conditions for no-information without
disturbance. An example is given in the following proposition:

\begin{proposition}
  An OPT satisfies no-information without disturbance upon
  $\widehat{X},\widehat{Y}$, with $X\subseteq D_{\Stset{(\rA)}}$,
  $Y\subseteq D_{\Cntset{(\rA)}}$, if there exists a pure state
  $\Psi\in\Ref_{\widehat{X}}$ that is faithful upon
  $\widehat{X},\widehat{Y}$. Similarly, an OPT satisfies
  no-information without disturbance upon $\widehat{X},\widehat{Y}$ if
  there exists an atomic effect $b\in\Ref_{\widehat{Y}}$ that is
  faithful upon $\widehat{X},\widehat{Y}$.
\end{proposition}

\begin{proof} We explicitly prove the case of faithful state, since
  that of faithful effect follows by analogy.  Given a system $\rA$,
  let $\Psi\in \Ref_{\widehat{X}}$ be pure and faithful upon
  $\widehat{X},\widehat{Y}$ (see Definition~\ref{def:faithuponD}). Now
  let the test $\{\tA_i\}_{i\in\rX } \in\Trnset(\rA)$ be
  non-disturbing upon $\widehat{X},\widehat{Y}$, namely
  $\sum_i\tA_i\tensor[_{\widehat{Y}}]{=}{_{\widehat{X}}} \tI _\rA$.
  Then, since $\Psi\in \Ref_{\widehat{X}}$ we have
  $\sum_i\tA_i\K\Psi=\K\Psi$, and since $\Psi$ is pure, there exists a
  set of probabilities $\{p_i\}_{i\in\rX }$ such that
  $\tA_i\K{\Psi}=p_i\K{\Psi}$. However, due to the faithfulness of
  $\Psi$, the map $\tA\mapsto \tA\K{\Psi}$ is injective upon
  $\widehat{X},\widehat{Y}$, and we conclude that
  $\tA_i\tensor[_{\widehat{Y}}]{=}{_{\widehat{X}}}p_i \tI_\rA$, and,
  by definition, the test $\{\tA_i\}_{i\in\rX } \in\Trnset(\rA)$ does
  not extract information upon $\widehat{X},\widehat{Y}$.
\end{proof}

As a corollary one has a sufficient condition for no-information
without disturbance with no restrictions on
inputs and outputs:

\begin{corollary}\label{prop:faith-no-info}
  An OPT satisfies no-information without disturbance if for every
  system $\rA$ there exists a pure state faithful for $\rA$ or an
  atomic effect faithful for $\rA$.
\end{corollary}

\section{Outlook on no-information without disturbance}\label{sec:VI}

In this last section we analyse the relation between
no-information without disturbance and other properties of operational
probabilistic theories. We focus on local discriminability (see
Definition~\ref{def:local-discriminability}) and purification (see
Definitions~\ref{def:states-purification} and
\ref{def:effects-purification}) that, being typical quantum features,
are commonly associated with no-information without disturbance. Here
instead we show that no-information without disturbance can actually
be satisfied independently of the above two properties.

\subsubsection*{No-information without disturbance vs. purification}

The following proposition proves that the probabilistic
theory~\cite{PhysRevA.75.032304,PhysRevA.71.022101,d2010testing,Short_2010,PhysRevLett.119.020401}
corresponding to the PR-boxes model, originally introduced in
Ref.~\cite{popescu1994quantum}, satisfies no-information without
disturbance.
\begin{proposition}\label{PR:prboxes}
  The PR-boxes theory satisfies no-information without disturbance.
\end{proposition}
\begin{proof}
  This can be proved in several ways. For example we show that any
  system of the theory allows for a reversible atomic transformation
  and then use Proposition~\ref{prop:niwd-rev}. The fact that any
  system has a reversible atomic transformation follows from the
  following three points. I) The reversible transformations of the
  elementary system $\rA$ of the theory (the convex set of normalized
  states of $\rA$ is represented by a square, and the set of
  reversible transformations of $\rA$ coincides with the set of
  symmetries of a square, the dihedral group of order eight $D_8$,
  containing four rotations and four reflections) are
  atomic~\cite{d2010testing}. II) From
  Refs.~\cite{PhysRevLett.104.080402,Al_Safi_2014} we know that the
  set of reversible maps of the $N$-partite system $\rA^{\otimes N}$
  is generated by local reversible operations plus permutations of the
  systems. Accordingly, the system $\rA^{\otimes N}$ allows for a
  multipartite reversible transformation $U_1\otimes U_2\otimes U_N$
  made of local reversible transformations $U_i$, $i=1,\ldots N$. III)
  Since PR-boxes satisfy local discriminability, the chosen
  multipartite transformation is atomic due to the atomicity of
  parallel composition (see Proposition~\ref{prop:atomic-parallel}).
\end{proof}

As a corollary one has that also PR-boxes with minimal tensor product
satisfy no-information without disturbance.

\begin{corollary}\label{prop:prboxes-minimal-tensor-product}
  The PR-boxes theory with minimal tensor product satisfies
  no-information without disturbance.
\end{corollary}

\begin{proof}
  Consider PR-boxes theory with minimal tensor
  product~\cite{BARNUM20113}. We remind that in a probabilistic
  framework, the minimal tensor product of state sets (or effect sets)
  of two systems is the operation that yields the set of states (or
  effects) of the composite system as containing only product states
  (or product effects) and their probabilistic mixtures. The PR-boxes
  theory with minimal tensor product has the same elementary system as
  the PR-boxes theory (with the convex set of normalized states
  represented by a square), but with composite systems constrained by
  the minimal tensor product for both states and effects. The proof
  that this probabilistic theory satisfies no-information without
  disturbance is as in Proposition~\ref{PR:prboxes} (one proves that
  every system has a reversible atomic transformation and
  no-information without disturbance follows from
  Proposition~\ref{prop:niwd-rev}).
\end{proof}

We can now establish the independence between no-information without
disturbance and purification (see also Figure~\ref{fig:sets}).

\begin{proposition}\label{prop:purification-niwd}
  No-information without disturbance and states or effects
  purification are independent.
\end{proposition}
\begin{proof}
We prove that:
\begin{enumerate}
\item\label{itm:pur1} Purification $\nRightarrow$ No-information
  without disturbance.
\item\label{itm:pur2} No-information without disturbance
  $\nRightarrow$ Purification
\end{enumerate}

\ref{itm:pur1}. Consider deterministic classical theory. This is
classical theory where the probabilities of outcomes in any test are
either 0 or 1 (see also Ref.~\cite{PhysRevA.81.062348}). One can
easily see that in this probabilistic theory all states are pure and
all effects are atomic. As a consequence the theory satisfies both
states and effects purification according to
Definitions~\ref{def:states-purification}
and~\ref{def:effects-purification}. On the other hand, in this theory
all information can be extracted without disturbance.

\ref{itm:pur2}. Consider PR-boxes theory with minimal tensor
product. On one hand this probabilistic theory satisfies
no-information without disturbance (see
Corollary~\ref{prop:prboxes-minimal-tensor-product}). On the other
hand, one can check that the theory satisfies neither states
purification, nor effects purification. First notice that, due to the
minimal tensor product prescription, multipartite pure states and
atomic effects are tensor products of local pure states and local
atomic effects, respectively. Therefore, given a multipartite pure
state (atomic effect) all its marginals are also pure states (atomic
effects). It follows that a state (effect) that is not pure (atomic)
cannot admit of any pure (atomic) dilation. Since the elementary system
of PR-boxes with minimal tensor product includes both non pure states
and non atomic effects, the theory does not satisfy states or effects
purification.
\end{proof}

\begin{figure}[t!]
\centering 
\includegraphics [width=\columnwidth]{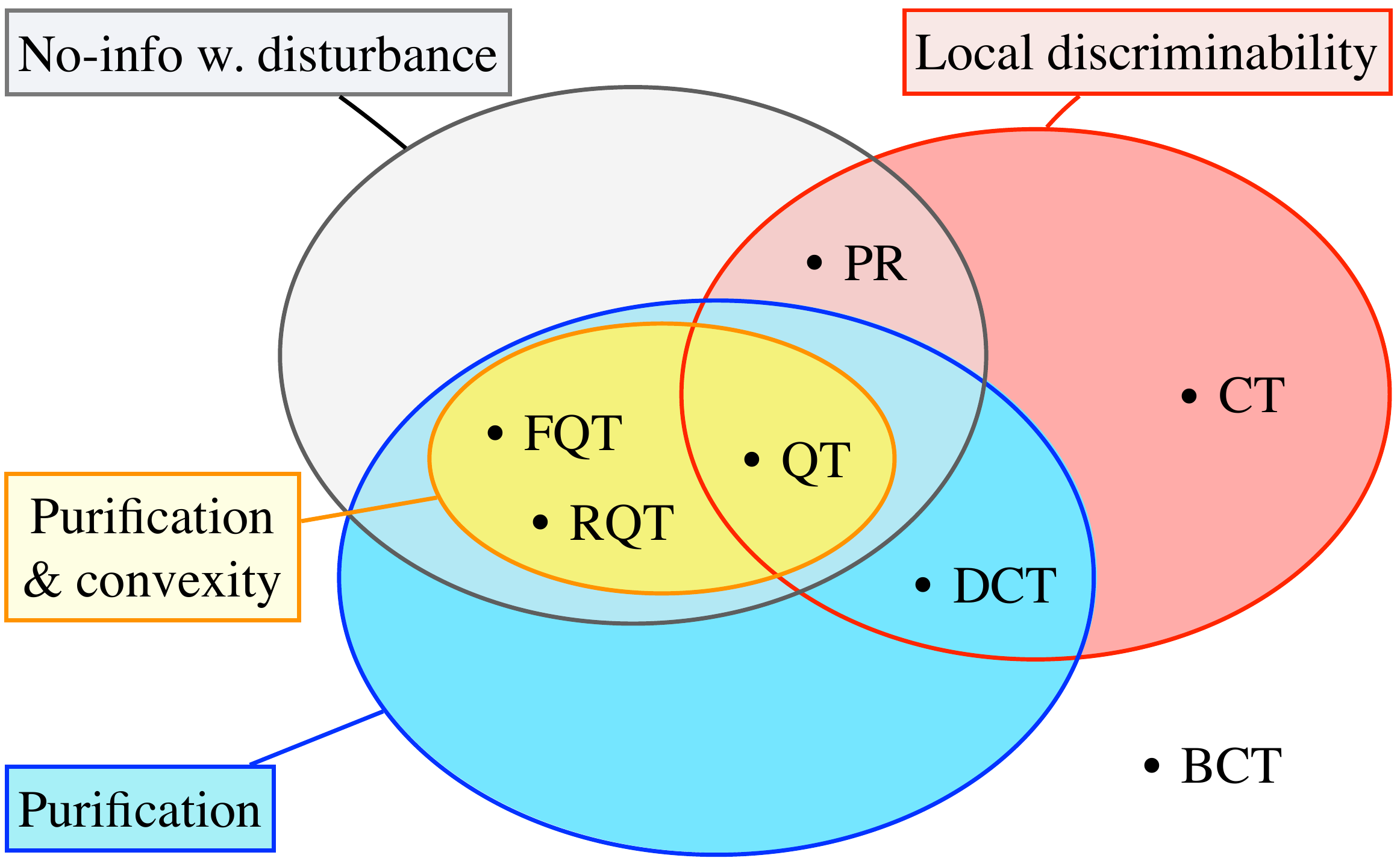}
\caption{Comparing OPTs that satisfy no-information without
  disturbance (grey set), local discriminability (red set) and states
  or effects purification (blue set). Quantum theory (QT) lies at the
  intersection of the three sets. As proved in
  Proposition~\ref{prop:purification-niwd} no-information without
  disturbance and purification are independent features. An example of
  OPT that satisfies no-information without disturbance but violates
  purification is the PR-boxes theory with minimal tensor product
  (PR). Moreover, PR-boxes satisfy local discriminability, providing a
  non-trivial intersection between local discriminability and
  no-information without disturbance in the absence of
  purification. On the other hand determinisctic classical theory
  (DCT) is an example of OPT that satisfies purification but violates
  no-information without disturbance. As proved in
  Proposition~\ref{prop:purification-niwd-2} the set of convex OPTs
  with purification is a proper subset of OPTs with no-information
  without disturbance. We observe that no-information without
  disturbance is also independent of local discriminability, as proved
  in Proposition~\ref{prop:local-disc-niwd}. Indeed classical theory
  (CT) satisfies only local discriminability while Fermionic quantum
  theory (FQT) and real quantum theory (RQT) satisfy only
  no-information without disturbance. Finally, it has been shown in
  Ref.~\cite{PhysRevA.102.052216} that there exist OPTs without local
  discriminability, that have all systems classical, thus retaining
  the possibility of extracting all the information without
  disturbance. An example is the {\em bilocal classical theory} (BCT)
  of the same Ref.~\cite{PhysRevA.102.052216}, which satisfies 2-local
  discriminability (see Definiton~\ref{def:ndisc}).}\label{fig:sets}
\end{figure}

While purification is not enough to imply no-information without
disturbance, in the next proposition we show that convex
OPTs~\footnote{Actually one might relax the convexity hypothesis with
  the following weaker condition: {\em If two states belong to a
    jointly perfectly discriminable set, there exists an element of
    the set that can be convexly mixed with both.} In particular, this
  is the case if some convex combination of the two states is a state
  of the theory.} that satisfy states or effect purification are a
subset of the OPTs with no-information without disturbance (actually
they are a proper subset, see also Fig.~\ref{fig:sets}). This provides
another useful sufficient condition for no-information without
disturbance (see also Figure~\ref{fig:sets}).

\begin{proposition}\label{prop:purification-niwd-2}
  A convex OPT with states purification or effects purification
  satisfies no-information without disturbance.
\end{proposition}
\begin{proof} We consider the case of states purification (see
  Definition~\ref{def:states-purification}), with the proof for
  effects purification (see Definition~\ref{def:effects-purification})
  following by analogy according to Remark~\ref{rem:states-effects}.

  Given a convex OPT with states purification suppose that it violates
  no-information without disturbance, namely there exists a system
  $\rA$ such that $\tI_\rA$ is not atomic. Then let
  $\tI_\rA=\sum_i\tA_i$, for some atomic non redundant test
  $\{\tA_i\}_{i\in\rX}\subseteq\Trnset(\rA)$.  Let us consider a mixed
  state $\Stset(\rA)\ni\K{\rho}=\sum_{i\in\rX}p_i\K{\sigma_i}$ with
  $p_i\K{\sigma_i}:=\tA_i\K{\rho}$, and $\{p_i\}_{i\in\rX}$ a
  probability distribution with $p_i>0$ $\forall i$. Then by
  Theorem~\ref{Thm:PROJ} we have
  $\tA_i\K{\sigma_j}=\delta_{ij}\K{\sigma_i}$. Since the theory allows
  for purification, let $\Psi\in\Stset{(\rA\rB)}$ be a purification of
  $\rho$ for deterministic effect $e\in\Cntset(\rB)$. Now, one one
  hand since the test $\{\tA_i\}_{i\in\rX}$ refines the identity, it
  is $\K{\Psi}_{\rA\rB}=\sum_{i\in\rX}\tA_i \K{\Psi}_{\rA\rB}$, and
  being $\Psi$ pure it must be
  $\tA_i \K{\Psi}_{\rA\rB}=q_i \K{\Psi}_{\rA\rB}$, with
  $\{q_i\}_{i\in\rX}$ a probability distribution. On the other hand,
  for every $i\neq j$ the marginals with deterministic effect
  $e\in\Cntset(\rB)$ of $\tA_i \K{\Psi}_{\rA\rB}$ and
  $\tA_j \K{\Psi}_{\rA\rB}$ are perfectly discriminable. But this
  contradicts the fact that $\tA_i \K{\Psi}_{\rA\rB}$ and
  $\tA_j \K{\Psi}_{\rA\rB}$ are both proportional to $\Psi$.
\end{proof}

\subsubsection*{No-information without disturbance vs. local
  discriminability}

Turning to the case of local discriminability (see Definition
\ref{def:local-discriminability}), we now show that it is independent
of no-information without disturbance (see also
Figure~\ref{fig:sets}).

\begin{proposition}\label{prop:local-disc-niwd}
  No-information without disturbance and local discriminability are
  independent.
\end{proposition}
\begin{proof} We prove that:
\begin{enumerate}
\item\label{itm:ld1} Local discriminability $\nRightarrow$ No-information without
  disturbance,

\item\label{itm:ld2} No-information without disturbance $\nRightarrow$
  Local discriminability.
\end{enumerate}

\ref{itm:ld1}. Classical theory satisfies local discriminability but
violates no-information without disturbance, since in this theory all
information can be extracted without disturbance.

\ref{itm:ld2}. Fermionic~\cite{Bravyi2002210,D_Ariano_2014,d2014feynman} and real quantum
theory~\cite{hardy2012limited,d2014feynman,D_Ariano_2014} both
violate local discriminability as proved in
Refs.~\cite{hardy2012limited,d2014feynman,D_Ariano_2014} where it is
also shown that they are both 2-local theories according to
Definition~\ref{def:ndisc}. On the other hand Fermionic and real quantum
theories are convex theories with states purification and, due to
Proposition~\ref{prop:purification-niwd-2}, both satisfy no-information
without disturbance.
\end{proof}

Finally, we observe that as a consequence of
Corollary~\ref{thm:class}, and subsequent
Remark~\ref{rem:classical-systems}, the classical theory of
information is the only theory with local discriminability in which
all the information can be extracted without disturbance. However, in
the absence of local discriminability, it is still possible to have
other theories where all the information can be extracted without
disturbance. This has been proved in Ref.~\cite{PhysRevA.102.052216}
where the authors describe an OPT whose systems of any dimension are
classical (and then violate no-information without disturbance), but
with a parallel composition that differs from the usual classical one,
leading to a violation of local discriminability, more precisely to a
2-local theory according to Definition~\ref{def:ndisc}. This theory is
interesting because it provides an example of OPT that violates
simultaneously no-information without disturbance, local
discriminability and purification (see bilocal classical theory (BCT)
in Figure~\ref{fig:sets}).

\section{Conclusions}\label{sec:VII}
We have analysed the interplay between information and disturbance for
a general operational probabilistic theory, considering the effect of
measurements also on correlations with the environment, differently
from the traditional approach focused only on the measured
system. Indeed, the two resulting notions of disturbance coincide only
in special cases, such as quantum theory, as well as every
theory that satisfies local discriminability. Our approach is
universal for any OPT, including theories without causality,
purification or convexity. In this setting we proved that the
atomicity of the identity transformation is an equivalent condition
for no-information without disturbance.

We have characterized the structure of theories where the identity is
not atomic, showing that in this case the information that can be
extracted without disturbance is ``classical'', in the sense that it
is sharable and repeatable. On the other hand, we have established
that every OPT entails information whose extraction requires
disturbance, with the only exception of theories with all systems
classical.

While no-information without disturbance is a consequence of convexity
along with purification (purification of states or of effects), we
proved that purification and no-information without disturbance are
independent.  Similarly, we have shown that no-information without
disturbance and local discriminability are independent properties.

Our results are expected to have immediate applicability to secure
key-distribution. Indeed, a physical theory including a system (or
even just a set of states of a system) that satisfies no-information
without disturbance can guarantee a private and reliable channel for
distributing messages. The idea of studying secure key-distribution in
a framework more general than the classical and the quantum ones has
been proposed in Refs.~\cite{barrett2007information,BARNUM20113}.  In
Ref.~\cite{barrett2007information} it has been conjectured that in
every theory that is not classical secure key-distribution is
possible. The present generalisation of no-information without
disturbance to arbitrary OPTs is a first step in proving such a conjecture.

\acknowledgments This publication was made possible thanks to the
financial support of Elvia and Federico Faggin Foundation.

 \bibliography{bibliography-doi-nourl}

\onecolumn\newpage
\appendix

\section{Transformations induced by events}\label{sec:transformations}
In the operational framework any event $\tA$ induces a map
between states. Consider for example the event
\begin{equation*}
  \Qcircuit @C=1em @R=.7em @! R { & \poloFantasmaCn \rA \qw &
    \gate{\tA} & \qw \poloFantasmaCn \rB &\qw}\ .
\end{equation*}
For every choice of ancillary system $\rC$, and for every state
$\Psi\in\Stset{(\rA\rC)}$, the event $\tA$ maps the state $\Psi$ to the
state given by the following circuit
\begin{equation*}
\begin{aligned}
  \Qcircuit @C=1em @R=.7em @! R {\multiprepareC{1}{\Psi}& \qw
    \poloFantasmaCn \rA &  \gate{\tA} & \qw \poloFantasmaCn \rB &\qw \\
    \pureghost\Psi & \qw \poloFantasmaCn \rC & \qw}
\end{aligned}\ .
\end{equation*}

Accordingly, while states and effects are linear functionals over each
other, we can always look at an event as a map between states
\begin{align}
  \tA:\K{\Psi}_{\rA\rC}\in\Stset{(\rA\rC)}\mapsto \tA\K{\Psi}_{\rA\rC}\in\Stset{(\rB\rC)},
\end{align}
(and similarly as a map between effects, from $\Cntset{(\rA\rC)}$ to
$\Cntset{(\rB\rC)}$). The map above can be linearly extended to a map
from $\Stset_\Reals{(\rA\rC)}$ to $\Stset_\Reals{(\rB\rC)}$ (we denote
the extended map with the same symbol $\tA$) and this extension is
unique. Indeed a linear combination of states of $\rA\rC$ is null, say
$\sum_i c_i \K{\Psi_i}$, if and only if $\sum_ic_i\SC{a}{\Psi_i}=0$
for any $a\in\Cntset{(\rB\rC)}$. Moreover, since
$\B{b}_\rB\tA\in\Cntset{(\rA\rC)}$ for every $b\in\Cntset{(\rB)}$,
then $\forall b\in\Cntset{(\rB)}$ we have that
$0=\B{b}_\rB\tA(\sum_i c_i\K{\Psi_i}_{\rA\rB})=\sum_i c_i
\B{b}_\rB\tA\K{\Psi_i}_{\rA\rB}=\B{b}_\rB\sum_i
c_i\tA\K{\Psi_i}_{\rA\rB}$,
and we finally get $\sum_i c_i \tA\K{\Psi_i}_{\rA\rB}=0$.

\section{No-information test}\label{App:no-info test}
We show that the weak condition~\ref{weak} for no-information test is
equivalent to one in Definition~\ref{def:noinfo}.

We focus now on no-information on the input (the case of
no-information on the output following by analogy) and show that
Eqs.~\eqref{eq:no-info-states-weak} and~\eqref{eq:no-info-states} are
equivalent. We first prove that Eq.~\eqref{eq:no-info-states} implies
Eq.~\eqref{eq:no-info-states-weak}. To this end we evaluate the left
hand side of Eq.~\eqref{eq:no-info-states-weak} using
Eq.~\eqref{eq:no-info-states}, namely
\begin{equation*}
\begin{aligned}
\sum_kp(j,i,k|\test{\Psi}_\rY,\test{\tA}_\rX,\test{C}_\rZ)=
\sum_k\B{c_k}\tA_i\K{\Psi_j}\\
=\B{e}\tA_i\K{\Psi_j}=p_i(e)\SC{f}{\Psi_j}.
\end{aligned}
\end{equation*}
Now we notice that $p_i(e)$, which is a probability distribution on
the outcomes of the test $\test{\tA}_\rX$, also depends on the
deterministic effect $e=\sum_k c_k$. Therefore $p_i(e)$ is a
probability distribution that depends on both the test
$\test{\tA}_\rX$ and on the observation test $\test{C}_\rZ$ while it
does not depend on the preparation test $\test{\Psi}_\rY$, exactly as
the probability distribution $r$ on the right hand side of
Eq.~\eqref{eq:no-info-states-weak}. Finally, we notice that
$\SC{f}{\Psi_j}$ is a probability distribution on the outcomes of the
preparation test $\test{\Psi}_\rY$, and the deterministic effect $f$
can depend on both tests $\test{\tA}_\rX$ and $\test{C}_\rZ$. We then
conclude that $\SC{f}{\Psi_j}$ is a probability distribution that
generally depends on all test $\test{\Psi}_\rY$, $\test{\tA}_\rX$ and
$\test{C}_\rY$, as the probability distribution $s$ on the right hand
side of Eq.~\eqref{eq:no-info-states-weak}. Now we check the other
implication, namely that Eq.~\eqref{eq:no-info-states-weak} implies
Eq.~\eqref{eq:no-info-states}. Due to
Eq.~\eqref{eq:no-info-states-weak}, one has
\begin{equation}\label{eq:comparing}
\begin{aligned}
&\B{e}\tA_i\K{\Psi_j}=
  \sum_k\B{c_k}\tA_i\K{\Psi_j}=  \sum_kp(j,i,k|\test{\Psi}_\rY,\test{\tA}_\rX,\test{C}_\rZ)\\
  &=r(i|\test{\test{\tA}}_\rX,\test{C}_\rZ)s(j|\test{\Psi}_\rY,\test{\tA}_\rX,\test{C}_\rZ).
\end{aligned}
\end{equation}
First we notice that summing over the index $i$ in the first and last member 
of Eq.~\eqref{eq:comparing} we get
$s(j|\test{\Psi}_\rY,\test{\tA}_\rX,\test{C}_\rZ)=\SC{f}{\Psi_j}$,
where $\B{f}=\sum_i\B{e}\tA_i$ is a deterministic effect depending on tests
$\test{\tA}_\rX,\test{C}_\rZ$ (indeed
$\sum_i\tA_i$ is a deterministic transformation). Therefore one has
$\B{e}\tA_i= r(i|\test{\test{\tA}}_\rX,\test{C}_\rZ) \B{f}$. Since on
the left hand side of the last identity the only dependence on the observation test
$\test{C}_\rZ$ is through the deterministic effect $e$, one finally gets
that $r(i|\test{\test{\tA}}_\rX,\test{C}_\rZ)$ is of the form $p_i(e)$,
as in Eq.~\eqref{eq:no-info-states}.

\section{Techical observation}\label{App:pulci}
Given a transformation $\tA\in\Trnset{(\rA)}$ such that for every
system $\rB$ there exists $p_\rB$ such that
\begin{equation}\label{p}
\tA\K{\rho}_{\rA\rB}=p_\rB\K{\rho}_{\rA\rB},
\end{equation}
then actually $p_\rB\equiv p$ cannot depend on the system
$\rB$. Indeed, choosing $\rho=\tau\otimes\sigma_\rB$ with arbitrary
$\tau\in\Stset(\rA)$ and normalised $\sigma\in\Stset(\rB)$, and
discarding system $\rB$ on both sides of Eq.~\eqref{p}
$\forall\rho\in\Stset(\rA\rB)$, one obtains
$\tA\K\tau=p_\rB\K\tau$. This clearly shows that $p_\rB\equiv p$.

\section{Techical lemma}\label{App:inclusion}
\begin{lemma}\label{lem:inclusion}
  Given a state $\rho\in\Stset(\rA)$ one has
  $\Ref_{D_\rho}\subseteq D_{\Ref_\rho}$, where $\Ref_{D_\rho}$
  denotes the union of the refinements of any state in $D_\rho$. Given
  an effect $a\in\Cntset(\rA)$ one has
  $\Ref_{D_a}\subseteq D_{\Ref_a}$, where $\Ref_{D_a}$ denotes the
  union of the refinements of any effect in $D_a$.
\end{lemma}
\begin{proof} Since the proof in the case of states and effects is
  exactly the same (see Remark~\ref{rem:states-effects}) we focus on
  the former. Consider a system $\rB$ and a state
  $\Psi\in\Stset{(\rA\rB)}$:
\begin{enumerate}
\item $\Psi\in D_{\Ref_\rho}$ iff
  $\B{e}_{\rB}\K{\Psi}_{\rA\rB}=\K{\sigma}_\rA$, with
  $\sigma\in\Ref_\rho$, and $e\in\Cntset(\rB)$ deterministic.
\item $\Psi\in \Ref_{D_{\rho}}$ iff there exists a state
  $\Omega\in\Stset{(\rA\rB)}$, with $\Omega\in D_\rho$ such that
  $\Psi\in\Ref_\Omega$, i.e.~there exists a refinement
  $\{\Psi_i\}_{i\in\rX }$ of $\Omega$ such that
  $\Psi\in \{\Psi_i\}_{i\in\rX }$.
\end{enumerate}
By hypothesis $\{\Psi_i\}_{i\in\rX }$ is a refinement of $\Omega$, namely 
\begin{equation}\label{KO}
\K{\Omega}_{\rA\rB}=\sum_{i\in\rX }\K{\Psi_i}_{\rA\rB},
\end{equation}
and $\Omega\in D_\rho$, namely $\B{e}_{\rB}\K{\Omega}_{\rA\rB}=\K{\rho}_\rA$.
Accordingly, marginalising  both sides of Eq. \eqref{KO} one has that  for every
$i\in\rX $, $\B{e}_{\rB}\K{\Psi_i}_{\rA\rB}=\K{\sigma_i}_\rA\in\Ref_\rho$. Since
$\Psi\in\{\Psi_i\}_{i\in\rX}$, this concludes the proof.
\end{proof}

\end{document}